\documentclass[aps, prx, twocolumn, superscriptaddress, 10pt]{revtex4-2}
\usepackage{amsthm}
\usepackage{amsmath,bm}
\usepackage{amssymb}
\usepackage{amsfonts}
\usepackage{graphicx}
\usepackage{txfonts}
\usepackage{xcolor}
\usepackage{braket}
\usepackage{bbm}
\usepackage{subcaption}
\usepackage{changes}
\usepackage{ragged2e}
\usepackage[colorlinks=true,linkcolor=blue,citecolor=blue,urlcolor=blue]{hyperref}
\usepackage[capitalise]{cleveref}
\usepackage{booktabs}
\usepackage{diagbox}

\newcommand{\ignore}[1]{} 

\newcounter{SaveEqnCntr}

\newcommand{\be}{\begin{equation}}
\newcommand{\ee}{\end{equation}}
\newcommand{\ba}{\begin{eqnarray}}
\newcommand{\ea}{\end{eqnarray}}

\def\>{\rangle}
\def\<{\langle}

\newtheorem{theorem}{Theorem}
\newtheorem{corollary}{Corollary}
\newtheorem{definition}{Definition}
\newtheorem{proposition}{Proposition}

\newtheorem{lemma}{Lemma}

\begin{document}
\title{Quantifying Margenau--Hill Nonclassicality}

\author{Sudip Chakrabarty}
\email{sudip27042000@gmail.com}
\affiliation{S. N. Bose National Centre for Basic Sciences, Block JD, Sector III, Salt Lake, Kolkata 700 106, India}

\begin{abstract}
Quasiprobability distributions offer a useful way of describing nonclassical features of quantum systems through their departure from classical probability theory. In this work, we investigate the quantification of nonclassicality associated with the Margenau--Hill quasiprobability (MHQ) distribution using its moments, without requiring reconstruction of the full distribution. First we introduce the logarithmic MHQ negativity as a quantifier of nonclassicality, and then derive a hierarchy of rigorous lower bounds in terms of low-order moments. Within the resulting hierarchy, the fourth-moment bound is the strongest among the bounds based on even moments. For qubits, we further derive a tight upper bound on the exact MHQ negativity and identify the corresponding extremal state. We also obtain an optimal positivity threshold based on the fourth moment for arbitrary pairs of qubit observables. Finally, we show that the relevant moments admit exact multicopy representations, providing a route to their estimation using interferometry or classical shadow techniques. Our results establish a moment-based framework for extracting quantitative information about MHQ negativity from a finite set of low-order observables.
\end{abstract}

\maketitle
\section{Introduction}
\label{s1} 
Nonclassicality is one of the central features that distinguishes quantum theory from classical physics. Phenomena such as superposition \cite{streltsov2017colloquium}, entanglement \cite{horodecki2009quantum}, nonlocality \cite{RevModPhys.86.419}, contextuality \cite{PhysRevA.71.052108,budroni2022kochen}, and measurement incompatibility~\cite{RevModPhys.95.011003} cannot, in general, be explained within a classical description and are responsible for many of the advantages offered by quantum technologies. For this reason, the detection and quantification of nonclassicality have attracted considerable attention in quantum information science \cite{lee1991measure,mari2011directly,vogel2014unified,innocenti2022nonclassicality, mallick2025efficient, chakrabarty2026operationaldetectionwignernegativity, probing_kirkwood}.

Quasiprobability representations provide a useful way of studying such nonclassical features by representing quantum states in a form analogous to classical probability distributions \cite{cahill1969density,ferrie2011quasi}. However, unlike ordinary probability distributions, quasiprobabilities are not required to satisfy all the Kolmogorov axioms, and this departure can reveal genuinely quantum behaviour that cannot be described by classical stochastic models \cite{veitch2012negative,tan2020negativity}. Some well-known examples are the Wigner distribution \cite{wigner1932quantum}, the Glauber--Sudarshan $P$ distribution \cite{sudarshan1963equivalence,glauber1963coherent}, and the Husimi $Q$ distribution \cite{husimi1940some}. Among these, the Wigner function has played a particularly important role, with its negativity being widely regarded as a signature of nonclassicality and linked to advantages in quantum computation, error correction, state distillation~\cite{kenfack2004negativity, spekkens2008negativity, mari2012positive, niset2009no,  PhysRevLett.89.137903}.

While phase-space quasiprobabilities are naturally suited for continuous-variable (CV) systems, most
of the modern quantum information protocols involve finite-dimensional systems and arbitrary
pairs of observables that are not restricted to canonical conjugate phase-space variables. This motivates the study of more general quasiprobability representations adapted to arbitrary pairs of incompatible observables. In this setting, the Kirkwood--Dirac (KD) quasiprobability distribution provides a natural generalization of phase-space quasiprobabilities. The KD distribution has recently attracted considerable attention because of its applications to weak measurements, quantum thermodynamics, metrology, quantum scrambling, and quantum computation~\cite{kirkwood1933quantum,dirac1945analogy,Budiyono_2023_correlation,budiyono_2024coherence,arvidsson2024properties,Lostaglio2023kirkwooddirac,Budiyono_2024_separation,budiyono_2025_entanglement, x819-898d}. However, since the KD distribution is generally complex-valued, its nonclassicality cannot be characterized solely through negativity. A particularly important related quasiprobability is the Margenau--Hill quasiprobability (MHQ) distribution~\cite{Margenau61}, obtained as the real part of the KD distribution. Being real-valued, the MHQ distribution admits a more direct interpretation: negative values quantify the departure from a classical joint-probability description associated with a given quantum state and a pair of observables.

Recent works have shown that MHQ negativity is not merely a mathematical curiosity, but has operational significance in quantum thermodynamics~\cite{ky8n-9bcy, PhysRevE.108.054109, Stepanyan2024energydensitiesin, PhysRevResearch.6.023280, PhysRevX.5.031038, D_az_2020, PRXQuantum.1.010309}, weak-measurements~\cite{PhysRevA.70.052115, Bizzarri_2025}, and many body systems~\cite{PRXQuantum.5.030201}. As these applications continue to grow, an important practical question naturally arises: how can one efficiently quantify MHQ negativity for an unknown quantum state? The straightforward approach is to reconstruct the entire quasiprobability distribution and evaluate its total negative weight. Such a strategy, however, becomes increasingly demanding with system size and typically requires complete measurement statistics or quasiprobability tomography~\cite{PhysRevLett.110.230602,PhysRevLett.110.230601,lundeen2011direct,PhysRevLett.108.070402,buscemi2013direct,langrenez2024convex}. This motivates the search for methods that certify or quantify MHQ negativity without reconstructing the full distribution.

In this work, we develop such a framework by exploiting the moments of the MHQ distribution, inspired by ideas from the classical moment problem. First, we introduce the logarithmic MHQ negativity, as a measure of MHQ negativity and derive rigorous lower bounds expressed solely in terms of a finite number of moments. Consequently, the proposed approach provides a compressed quantification of MHQ negativity while bypassing full quasiprobability reconstruction. Beyond the general framework, we establish several results that strengthen the practical applicability of the method. We identify the optimal fourth-moment lower bound within the proposed hierarchy, and obtain a tight universal bound on the exact qubit negativity together with an exact measurement-dependent certification threshold for qubit systems. We further extend the framework for representative CV states. Finally, we show that the required moments admit exact multicopy representations, allowing them to be estimated in principle through ancilla-assisted interferometric visibility estimation or nonlinear classical-shadow protocols.

The remainder of the paper is organized as follows. Section~\ref{s2} introduces the MHQ distribution and the corresponding moment framework. 
In Sec.~\ref{s3}, we define logarithmic MHQ negativity and establish its basic properties. 
Section~\ref{s4} develops the moment-based quantification, and studies important properties of the quantifier.
Section~\ref{s5} establishes bounds on qubit MHQ negativity, including a universal closed-form bound on the exact negativity and an optimal threshold for the fourth-moment quantifier.
An extension of the quantification strategy to continuous variable systems is discussed in Section~\ref{s6}. 
Section~\ref{s7} discusses multicopy representations of the MHQ moments and their estimation protocols.
Finally, Sec.~\ref{s8} concludes with a discussion and future directions.

\section{Preliminaries}
\label{s2}

We first introduce the MHQ distribution associated with a quantum state and a pair of observables, and discuss its basic properties. We then turn to its moments, which are the key tools for the quantification of MHQ negativity developed in this work.

\subsection{Margenau--Hill quasiprobability distribution}
\label{s2A}

Let $\rho\in\mathcal{D}(\mathbb{H})$ be a density operator
on a $d$-dimensional Hilbert space $\mathbb{H}$. Consider
two observables
\begin{equation}
A=\sum_i a_i\Pi_i^a,
\qquad
B=\sum_j b_j\Pi_j^b,
\end{equation}
where
\(
\Pi_i^a=\ket{a_i}\bra{a_i}
\)
and
\(
\Pi_j^b=\ket{b_j}\bra{b_j}
\)
are rank-one spectral projectors. The MHQ distribution associated with the triple
$(\rho,A,B)$ is defined as~\cite{Margenau61}
\begin{equation}
M_{ij}(\rho;A,B)
=
\frac{1}{2}
\operatorname{Tr}
\left[
\left(
\Pi_i^a\Pi_j^b
+
\Pi_j^b\Pi_i^a
\right)
\rho
\right].
\label{eq:MHQ_definition}
\end{equation}
Equivalently,
$M_{ij}(\rho;A,B)
=
\operatorname{Re}
\operatorname{Tr}
\left(
\Pi_j^b\Pi_i^a\rho
\right),$
so that the MHQ distribution can be understood as the real part of the corresponding KD distribution. The symmetrized operator acting on $\rho$ in
Eq.~\eqref{eq:MHQ_definition} is Hermitian. Hence every $M_{ij}$ is real. The MHQ distribution reproduces the correct marginal statistics associated with the measurements of $A$ and $B$. These properties follow directly from the definition and are summarized below. 
\begin{align}
\sum_{i,j}M_{ij}(\rho;A,B)
&=1, \nonumber
\\
\sum_j M_{ij}(\rho;A,B)
&=
\operatorname{Tr}(\Pi_i^a\rho), \nonumber
\\
\sum_i M_{ij}(\rho;A,B)
&=
\operatorname{Tr}(\Pi_j^b\rho).
\label{eq:MHQ_normalization}
\end{align}
Thus, although the MHQ distribution need not be a genuine
joint probability distribution, its marginals reproduce the
Born-rule statistics of the individual measurements.
A particularly simple situation arises when the two observables commute. In this case, the projectors $\Pi_i^a$ and $\Pi_j^b$ can be simultaneously diagonalized, and after a common relabeling of the shared eigenbasis, one finds
\begin{equation}
    M_{ij}(\rho; A,B) = \delta_{ij}\,\mathrm{Tr}(\Pi_i^a \rho),
\end{equation}
so that the MHQ distribution reduces to a classical probability distribution. Negative entries can therefore occur
only for incompatible measurements, although incompatibility
alone does not guarantee negativity for every state.
The real-valued structure of the MHQ distribution makes its
negative entries directly interpretable as a failure of a
classical joint-probability description for the chosen pair of
observables. Throughout this work, the term
\emph{MHQ negativity} refers specifically to this
state- and measurement-dependent nonpositivity.

\subsection{Moments of the MHQ distribution}
\label{s2B}

The moment sequence of a real distribution provides a compact
description of its global structure. For the finite MHQ
distribution, we define the following quantities.

\begin{definition}[MHQ moments]
The $n$th-order moment of the MHQ distribution associated
with $(\rho,A,B)$ is
\begin{equation}
m_n(\rho;A,B)
=
\sum_{i,j}
\left[
M_{ij}(\rho;A,B)
\right]^n,
\qquad
n\in\mathbb{N}.
\label{eq:MHQ_moments}
\end{equation}
\end{definition}

Normalization implies
$m_1(\rho;A,B)=1.$
Even-order moments depend only on the magnitudes of the MHQ
elements, whereas odd-order moments also retain information
about their signs. For later use, we also introduce the moments of the
absolute-valued distribution.
\begin{definition}[Absolute MHQ moments]
The $n$th-order absolute moment is defined as
\begin{equation}
\widetilde{m}_n(\rho;A,B)
=
\sum_{i,j}
\left|
M_{ij}(\rho;A,B)
\right|^n,
\qquad n>0
\label{eq:absolute_MHQ_moments}
\end{equation}
\end{definition}
The first absolute moment satisfies
\begin{equation}
\widetilde{m}_1
=
\sum_{i,j}|M_{ij}|
=
1+2
\sum_{M_{ij}<0}|M_{ij}|,
\label{eq:absolute_first_moment}
\end{equation}
where the second equality follows from
Eq.~\eqref{eq:MHQ_normalization}. Consequently,
\begin{equation}
\widetilde{m}_1=1
\quad\Longleftrightarrow\quad
M_{ij}\geq0
\ \ \forall\,i,j,
\label{eq:m1_positivity}
\end{equation}
whereas
\(
\widetilde{m}_1>1
\)
is equivalent to MHQ negativity. Moreover,
\begin{equation}
\widetilde{m}_{2n}=m_{2n},
\label{eq:absolute_even_equality}
\end{equation}
which will allow us to construct quantitative bounds using
ordinary even-order moments alone.

The relevance of moment matrices follows from the classical
moment problem. A real sequence $\{\mu_k\}_{k\geq0}$ is a
moment sequence of a positive Borel measure on
$\mathbb{R}$ only if the associated Hankel matrices
\begin{equation}
\big[H_p(\boldsymbol{\mu})\big]_{kl}
=
\mu_{k+l},
\qquad
k,l\in\{0,\ldots,p\},
\label{eq:classical_Hankel}
\end{equation}
are positive semidefinite at every order. Indeed, for any real
vector
\(
\boldsymbol{c}=(c_0,\ldots,c_p)^{T}
\),
\begin{equation}
\boldsymbol{c}^{T}
H_p(\boldsymbol{\mu})
\boldsymbol{c}
=
\int_{\mathbb{R}}
\left(
\sum_{k=0}^{p}c_kx^k
\right)^2
d\xi(x)
\geq0.
\label{eq:Hankel_quadratic_form}
\end{equation}

The same reasoning applies to a nonnegative MHQ distribution.
In this case, the MHQ elements themselves may be viewed as
the support points of a positive discrete measure whose
weights are also given by the MHQ elements. This leads to the
following necessary condition.

\begin{theorem}[Moment-matrix condition]
\label{thm:MHQ_Hankel}
If the MHQ distribution $(M (\rho))$ of a quantum state $\rho$ with respect to bases $\{\ket{a_i}\}$ and $\{\ket{b_j}\}$ with $i,j =1,2,...,d$, is positive, then
\begin{equation}
       H_{p}(\mathbf{m}) \succeq 0 . \label{Hankelmatrix}
    \end{equation}
Here, $[H_{p}(\mathbf{m})]_{kl} = m_{k+l+1}$ for $k,l \in \{0,1,...,p\}$, $p \in \mathbb{N}$ and $ \mathbf{m}= (m_1, m_2, ..., m_{2p+1})$ are the corresponding MHQ moments defined in Eq.~\eqref{eq:MHQ_moments}. Consequently, 
if $H_{p}(\mathbf{m}) \nsucceq  0$ for some $p$, the MHQ distribution is thereby certified to be nonclassical.
\end{theorem}

For completeness, the proof is given in Appendix~\ref{app:1}. Theorem~\ref{thm:MHQ_Hankel} provides a family of sufficient
tests for MHQ negativity. In the following sections, we use the same moment data for a different purpose: constructing
certifiable lower bounds on the logarithmic MHQ negativity without reconstructing the complete quasiprobability distribution.

The use of moments as diagnostic tools has a long history in both classical and quantum settings. As already mentioned, in classical probability theory, the moment problem concerns the characterization of distributions from their moments, with positivity conditions on associated moment matrices providing necessary and sufficient criteria for \emph{existence}. Whereas in quantum information, moment-based methods have been successfully employed in detection of quantum entanglement~\cite{yu2021optimal, gray2018machine, PhysRevLett.125.200501, mallick2025higher, mukherjee2025efficient, PhysRevLett.109.130502, Tarabunga2026quantifyingmixed}, Wigner negativity~\cite{mallick2025efficient, chakrabarty2026operationaldetectionwignernegativity}, KD nonpositivity~\cite{probing_kirkwood}, and many other resources~\cite{PhysRevA.109.022247, chakrabarty2026detectionquantumimaginarityusing, x56c-2yp1}. 
Our earlier work~\cite{probing_kirkwood} introduced moment-based criteria for detecting KD nonpositivity. A similar approach is used in developing detection criteria for MHQ negativity in Theorem~\ref{thm:MHQ_Hankel}. However, in general, these criteria do not provide any information about the magnitude of the negative weight.
This present work addresses the distinct quantitative problem rather than mere detection: estimating the total negative weight of the MHQ distribution from a finite set of moments.

\section{Logarithmic quantification of Margenau--Hill
nonclassicality}
\label{s3}
While the presence of negativity in the MHQ distribution already signals nonclassicality for a specified state and pair of measurements, a more refined question naturally arises: \emph{how much nonclassicality is present in the MHQ distribution?} A quantitative answer to this question is essential, particularly in applications where MHQ negativity acts as a resource.

For the MHQ distribution
\(
\{M_{ij}(\rho;A,B)\}
\),
the total negative weight is defined as
\begin{equation}
\mathcal{N}_{\mathrm{MHQ}}(\rho;A,B)
= \sum_{M_{ij}<0}
\left|M_{ij}(\rho;A,B)\right|.
\label{eq:MHQ_negative_weight}
\end{equation}
Since the MHQ distribution is normalized,
\begin{equation}
\sum_{i,j}M_{ij}(\rho;A,B)=1,
\end{equation}
its first absolute moment satisfies
\begin{equation}
\widetilde{m}_1(\rho;A,B)
=
\sum_{i,j}
\left|M_{ij}(\rho;A,B)\right|
=
1+
2\mathcal{N}_{\mathrm{MHQ}}(\rho;A,B).
\label{eq:MHQ_l1_negative_weight}
\end{equation}
Thus,
\(
\widetilde{m}_1=1
\)
if and only if every MHQ element is nonnegative, whereas
\(
\widetilde{m}_1>1
\)
is equivalent to MHQ negativity.

We use the logarithm of this absolute weight as our
quantifier.

\begin{definition}[Logarithmic MHQ negativity]
\label{def:log_MHQ}
For a state $\rho$ and two projective observables $A$ and
$B$, the logarithmic MHQ negativity is defined as
\begin{equation}
L_{A,B}(\rho)
=
\ln
\left[
\sum_{i,j}
\left|M_{ij}(\rho;A,B)\right|
\right].
\label{eq:log_MHQ_definition}
\end{equation}
Equivalently,
\begin{equation}
L_{A,B}(\rho)
=
\ln
\left[
1+
2\mathcal{N}_{\mathrm{MHQ}}(\rho;A,B)
\right].
\label{eq:log_MHQ_negative_weight}
\end{equation}
Throughout this work, $\ln$ denotes the natural logarithm.
\end{definition}

The notation \(L_{A,B}(\rho)\) emphasizes that MHQ
negativity is not a property of the state alone: it also
depends on the chosen measurement pair. The quantity in
Eq.~\eqref{eq:log_MHQ_definition} is a faithful quantifier of
the total negative weight of the corresponding MHQ
distribution. We next collect several properties that will be
useful in the subsequent analysis.

\begin{proposition}
\label{prop:log_MHQ_properties}
Let $\rho\in\mathcal{D}(\mathbb{H})$, where
$\dim\mathbb{H}=d$, and let $A$ and $B$ be rank-one
projective observables with $d$ outcomes. Then
$L_{A,B}(\rho)$ satisfies the following properties.

\begin{enumerate}

\item \emph{Faithfulness and non-negativity:}
\begin{equation}
L_{A,B}(\rho)\geq0,\;\;
L_{A,B}(\rho)=0
\;\Longleftrightarrow\;
M_{ij}(\rho;A,B)\geq0
\quad
\forall\,i,j.
\label{eq:log_MHQ_faithfulness}
\end{equation}

\item \emph{Unitary covariance:} For every unitary $U$,
\begin{equation}
L_{A,B}(U\rho U^\dagger)
=
L_{U^\dagger A U,U^\dagger B U}(\rho).
\label{eq:log_MHQ_covariance}
\end{equation}

\item \emph{Invariance under outcome relabeling:}
Independent permutations of the outcome labels of either
$A$ or $B$ leave \(L_{A,B}(\rho)\) unchanged.

\item \emph{Dimension-dependent upper bound:}
\begin{equation}
L_{A,B}(\rho)\leq\ln d.
\label{eq:log_MHQ_upper_bound}
\end{equation}

\item \emph{Necessity of measurement incompatibility:}
If the spectral projectors of $A$ and $B$ commute pairwise,
then
$L_{A,B}(\rho)=0,
\;
\forall \rho.$
Consequently,
\begin{equation}
L_{A,B}(\rho)>0
\quad\Longrightarrow\quad
[A,B]\neq0.
\label{eq:log_MHQ_incompatibility}
\end{equation}

\item \emph{Necessity of coherence relative to $A$:}
\begin{equation}
[\rho,\Pi_i^a]=0,
\;\;
\forall\,i
\; \; \Longrightarrow \;\; L_{A,B}(\rho)=0.
\label{eq:log_MHQ_incoherent}
\end{equation}
For a nondegenerate observable $A$, condition
\eqref{eq:log_MHQ_incoherent} is equivalent to
\(
[\rho,A]=0
\).

\item \emph{Vanishing after complete $A$-dephasing:}
Let
\begin{equation}
\Delta_A(\rho)
=
\sum_i
\Pi_i^a\rho\Pi_i^a.
\label{eq:A_dephasing}
\end{equation}
Then
\begin{equation}
L_{A,B}\!\left(\Delta_A(\rho)\right)=0
\leq
L_{A,B}(\rho).
\label{eq:log_MHQ_dephasing}
\end{equation}

\end{enumerate}
\end{proposition}

The proofs for these properties are provided in
Appendix~\ref{app:log_MHQ_properties}. These properties
clarify the scope of the quantifier. In particular, nonzero
\(L_{A,B}(\rho)\) requires both measurement incompatibility
and coherence of the state relative to the measurement
basis, although neither condition alone is sufficient to
guarantee MHQ negativity.

\section{Certifiable lower bounds for logarithmic MHQ negativity}
\label{s4}

The logarithmic MHQ negativity
\(
L_{A,B}(\rho)
\)
depends on the full set of quasiprobability elements
\(
\{M_{ij}(\rho;A,B)\}
\).
Its direct evaluation therefore requires full knowledge about the MHQ distribution.
This naturally motivates the following question: \emph{can one estimate MHQ negativity without reconstructing the full MHQ distribution?}

In this section, we show that a finite collection of moments can nevertheless
provide rigorous lower bounds on the logarithmic MHQ negativity. The construction is based on a probability distribution
obtained by normalizing the squared magnitudes of the MHQ
elements. We now introduce a normalized distribution
\begin{equation}
f_{ij}
=
\frac{|M_{ij}(\rho;A,B)|^2}
{\widetilde m_2(\rho;A,B)},
\label{eq:normalized_squared_MHQ}
\end{equation}
where
$\widetilde m_2
=
\sum_{i,j}|M_{ij}|^2,$ the second absolute MHQ moment.
Since the MHQ distribution is normalized, it cannot vanish
identically, and hence
\(
\widetilde m_2>0
\).
Consequently,
\begin{equation}
f_{ij}\geq0,
\qquad
\sum_{i,j}f_{ij}=1.
\end{equation}

For a probability distribution \(f=\{f_{ij}\}\), the Rényi
entropy of order \(\alpha>0\), with \(\alpha\neq1\), is
\begin{align}
R_{\alpha}(f)
&=
\frac{1}{1-\alpha}
\ln
\left(
\sum_{i,j}f_{ij}^{\alpha}
\right)
\label{eq:Renyi_entropy}\\
&=
\frac{1}{1-\alpha}
\ln
\left(
\frac{\widetilde m_{2\alpha}}
{\widetilde m_2^{\,\alpha}}
\right).
\label{eq:Renyi_MHQ_moments}
\end{align}

The order-\(1/2\) Rényi entropy is directly related to the
logarithmic MHQ negativity. Indeed,
\begin{align}
R_{1/2}(f)
&=
2\ln
\left(
\sum_{i,j}\sqrt{f_{ij}}
\right)
\nonumber\\
&=
2\ln
\left(
\frac{\widetilde m_1}
{\sqrt{\widetilde m_2}}
\right)
\nonumber\\
&=
2L_{A,B}(\rho)
-
\ln\widetilde m_2.
\label{eq:Renyi_half_log_negativity}\\
\Rightarrow \; L_{A,B}(\rho)
&=
\frac12
\left[
R_{1/2}(f)
+
\ln\widetilde m_2
\right].
\label{eq:L_from_Renyi_half}
\end{align}

Before proceeding further, we discuss one important property of Rényi entropies, which will be necessary to develop the lower bound. We write this property in the following lemma.

\begin{lemma}\label{lemma:renyi}
Let $R_r(f)$ denote the Rényi entropy associated with the probability distribution $\{f_{ij}\}$. Then for $r < s$,
\begin{equation}
    R_r(f) \ge R_s(f). \label{eq:Renyi_monotonicity}
\end{equation}
\end{lemma}
The proof is provided in Appendix~\ref{app:renyi_lemma}. We now define the moment-based quantifier.

\begin{definition}[Moment-based lower-bound functional]
\label{def:Lr_MHQ}
For \(r>2\), define
\begin{align}
L_{r;A,B}(\rho)
&=
\frac{1}{2-r}
\left[
\ln\widetilde m_r
+
(1-r)\ln\widetilde m_2
\right]
\label{eq:Lr_definition}
\\
&=
\frac12
\left[
R_{r/2}(f)
+
\ln\widetilde m_2
\right].
\label{eq:Lr_Renyi_form}
\end{align}
\end{definition}

The second equality follows directly from
Eq.~\eqref{eq:Renyi_MHQ_moments} with
\(
\alpha=r/2
\).
We refer to \(L_{r;A,B}\) as a lower-bound functional rather
than a negativity measure, since it can take negative values
even when the underlying MHQ distribution is nonnegative. The usefulness of this construction follows from Rényi-entropy
monotonicity.

\begin{theorem}[Moment-based lower bound]
\label{thm:Lr_lower_bound}
For every \(r>2\),
\begin{equation}
L_{r;A,B}(\rho)
\leq
L_{A,B}(\rho).
\label{eq:Lr_lower_bound}
\end{equation}
\end{theorem}

\begin{proof}
Since \(r>2\), one has
\(
r/2>1/2
\).
Equation~\eqref{eq:Renyi_monotonicity} therefore implies
\begin{equation*}
R_{r/2}(f)
\leq
R_{1/2}(f).
\end{equation*}
Adding
\(
\ln\widetilde m_2
\)
to both sides and dividing by \(2\), we obtain
\begin{equation*}
\frac12
\left[
R_{r/2}(f)
+
\ln\widetilde m_2
\right]
\leq
\frac12
\left[
R_{1/2}(f)
+
\ln\widetilde m_2
\right].
\end{equation*}
Using Eqs.~\eqref{eq:Lr_Renyi_form} and
\eqref{eq:L_from_Renyi_half} we get,
\begin{equation*}
L_{r;A,B}(\rho)
\leq
L_{A,B}(\rho).
\end{equation*}
This completes the proof.
\end{proof}

The same result may alternatively be obtained from the
interpolation inequality
$\widetilde m_2^{\,r-1}
\leq
\widetilde m_1^{\,r-2}
\widetilde m_r,$
but the Rényi-entropy formulation makes both the lower bound
and its ordering with \(r\) transparent. However, for completeness, we provide the alternative proof in Appendix~\ref{app:alternate_proof}.

Theorem~\ref{thm:Lr_lower_bound} shows that $L_r(\rho;A,B)$ provides a lower bound to the logarithmic MHQ negativity while depending only on two moments of the distribution avoiding the need for full MHQ reconstruction. We next show that these quantities also act as witnesses of MHQ negativity.

\begin{corollary}\label{corollary:quantitative_witness}
Let $M_{ij}(\rho;A,B)$ be the MHQ distribution corresponding to a quantum state $\rho$, and two observables $A$ and $B$. Then the following statements hold:
\begin{enumerate}
    \item If $M_{ij}(\rho;A,B) \ge 0$ for all $i,j$, then $L_{r;A,B}(\rho) \le 0$.
    
    \item If $L_{r;A,B}(\rho) > 0$, then the MHQ distribution is not a probability distribution.
\end{enumerate}
\end{corollary}

\begin{proof}
From Theorem~\ref{thm:Lr_lower_bound} we have,
$L_{r;A,B}(\rho)
\leq
L_{A,B}(\rho).$
Since \(
L_{A,B}(\rho)=0
\) whenever the MHQ distribution is classical, we also have $L_{r;A,B}(\rho)\leq 0$. For the second statement, whenever $L_{r;A,B}(\rho)> 0$, it also implies
\(
L_{A,B}(\rho)>0
\),
which, by the faithfulness of logarithmic MHQ negativity,
is equivalent to the presence of a negative MHQ element.
\end{proof}

However the converse does not generally hold: an MHQ-negative
distribution may satisfy
\(
L_{r;A,B}\leq0
\)
for some $r$. Thus, \(L_{r;A,B}>0\) is a \emph{sufficient} condition for MHQ negativity.

\subsection{Even moment quantifiers} 
Absolute MHQ moments involve nonlinear functions of the individual quasiprobability elements and are therefore not, in general, represented by the ordinary multicopy power observables introduced in Sec.~\ref{s7}.
However, a particularly useful simplification occurs for even values of
\(r\). Since the MHQ distribution is real,
\begin{equation}
\widetilde m_{2n}
=
\sum_{i,j}|M_{ij}|^{2n}
=
\sum_{i,j}M_{ij}^{2n}
=
m_{2n}.
\label{eq:even_absolute_ordinary}
\end{equation}
Consequently,
\begin{equation}
L_{2n;A,B}(\rho)
=
\frac{1}{2-2n}
\left[
\ln m_{2n}
+
(1-2n)\ln m_2
\right],
\qquad
n\geq2,
\label{eq:even_L2n}
\end{equation}
depends only on the ordinary even-order MHQ moments. These
moments admit multicopy operator representations, as discussed
in Sec.~\ref{s7}, and therefore do not require direct access to
the signs of the individual MHQ elements.

Within this even-order family, the lowest nontrivial member
provides the largest lower bound.

\begin{proposition}[Ordering of the even-moment bounds]
\label{prop:L4_optimal_family}
For integers \(2\leq n<m\),
\begin{equation}
L_{2n;A,B}(\rho)
\geq
L_{2m;A,B}(\rho).
\label{eq:even_bound_ordering}
\end{equation}
In particular,
\begin{equation}
L_{4;A,B}(\rho)
\geq
L_{6;A,B}(\rho)
\geq
L_{8;A,B}(\rho)
\geq\cdots.
\label{eq:L4_hierarchy}
\end{equation}
\end{proposition}

\begin{proof}
Since \(n<m\), monotonicity of the Rényi entropy gives
$R_n(f)\geq R_m(f)$.
Using Eq.~\eqref{eq:Lr_Renyi_form},
\begin{align*}
L_{2n;A,B}(\rho)
&=
\frac12
\left[
R_n(f)+\ln\widetilde m_2
\right],
\\
L_{2m;A,B}(\rho)
&=
\frac12
\left[
R_m(f)+\ln\widetilde m_2
\right].
\end{align*}
The term
\(
\ln\widetilde m_2
\)
is common to both expressions, and therefore
\begin{equation*}
L_{2n;A,B}(\rho)
\geq
L_{2m;A,B}(\rho).
\end{equation*}
This completes the proof.
\end{proof}

The first member of this family is
\begin{equation}
L_{4;A,B}(\rho)
=
\frac32\ln m_2
-
\frac12\ln m_4.
\label{eq:L4_explicit}
\end{equation}
Proposition~\ref{prop:L4_optimal_family} shows that \(L_4\)
is the strongest lower bound among the particular family
\(\{L_{2n}\}_{n\geq2}\). However, we clarify that this statement should not be
interpreted as a global optimality result: a stronger
certificate may be obtained from a different function of the
same moments, from several moments used jointly, or from
additional information about the state and measurements.

The results above establish a direct connection between
finite moment data and the total MHQ negativity. The exact
quantity \(L_{A,B}(\rho)\) measures the full absolute weight
of the quasiprobability distribution, whereas
\(L_{2n;A,B}(\rho)\) provides a certified lower bound using
only the moments \(m_2\) and \(m_{2n}\). In the next
subsection, we exploit the geometry of qubit MHQ
distributions to sharpen the generic witness
\(L_{4;A,B}\) for dichotomic measurements.

\section{Bounds on qubit MHQ negativity}
\label{s5}

The condition
\(
L_4(\rho;A,B)>0
\)
is sufficient to certify negativity of the MHQ distribution. For qubits, this condition can
be sharpened considerably by determining the largest value of
\(L_4\) that is compatible with an everywhere nonnegative MHQ
distribution for a fixed pair of measurements. A general qubit state can be written as
\begin{equation}
\rho
=
\frac{1}{2}
\left(
\mathbb{I}
+
\boldsymbol{r}\cdot\boldsymbol{\sigma}
\right),
\qquad
|\boldsymbol{r}|\leq 1,
\label{eq:qubit_bloch_state}
\end{equation}
where
\(
\boldsymbol{r}\in\mathbb{R}^{3}
\)
is the Bloch vector and
\(
\boldsymbol{\sigma}
=
(\sigma_x,\sigma_y,\sigma_z)
\).

We consider two dichotomic projective observables
\begin{equation}
A
=
\widehat{\boldsymbol{a}}\cdot\boldsymbol{\sigma},
\qquad
B
=
\widehat{\boldsymbol{b}}\cdot\boldsymbol{\sigma},
\label{eq:qubit_observables}
\end{equation}
with outcomes \(s,t\in\{+1,-1\}\) and projectors
\begin{equation}
\Pi_s^A
=
\frac{1}{2}
\left(
\mathbb{I}
+
s\,\widehat{\boldsymbol{a}}
\cdot\boldsymbol{\sigma}
\right),
\qquad
\Pi_t^B
=
\frac{1}{2}
\left(
\mathbb{I}
+
t\,\widehat{\boldsymbol{b}}
\cdot\boldsymbol{\sigma}
\right).
\label{eq:qubit_projectors}
\end{equation}
The corresponding MHQ elements are
\begin{equation}
M_{st}(\rho;A,B)
=
\frac{1}{2}
\operatorname{Tr}
\left[
\left(
\Pi_s^A\Pi_t^B
+
\Pi_t^B\Pi_s^A
\right)
\rho
\right].
\label{eq:qubit_MHQ_definition}
\end{equation}
Using
\begin{equation}
\left(
\widehat{\boldsymbol{a}}\cdot\boldsymbol{\sigma}
\right)
\left(
\widehat{\boldsymbol{b}}\cdot\boldsymbol{\sigma}
\right)
=
\left(
\widehat{\boldsymbol{a}}
\cdot
\widehat{\boldsymbol{b}}
\right)\mathbb{I}
+
i
\left(
\widehat{\boldsymbol{a}}
\times
\widehat{\boldsymbol{b}}
\right)
\cdot\boldsymbol{\sigma},
\end{equation}
one obtains
\begin{equation}
M_{st}
=
\frac{1}{4}
\left[
1
+
s\,\widehat{\boldsymbol{a}}\cdot\boldsymbol{r}
+
t\,\widehat{\boldsymbol{b}}\cdot\boldsymbol{r}
+
st\,
\widehat{\boldsymbol{a}}\cdot\widehat{\boldsymbol{b}}
\right].
\label{eq:general_qubit_MHQ}
\end{equation}
The component of \(\boldsymbol{r}\) orthogonal to the plane
spanned by
\(\widehat{\boldsymbol{a}}\) and
\(\widehat{\boldsymbol{b}}\)
does not contribute to Eq.~\eqref{eq:general_qubit_MHQ}.
We may therefore choose coordinates such that
\begin{equation}
A=\sigma_z,
\qquad
B(\phi)
=
\cos(2\phi)\,\sigma_z
+
\sin(2\phi)\,\sigma_x,
\label{eq:canonical_qubit_measurements}
\end{equation}
where
$\widehat{\boldsymbol{a}}
\cdot
\widehat{\boldsymbol{b}}
=
\cos(2\phi),
\;
0\leq\phi\leq{\pi}/{2}.$
It is then sufficient to take
\begin{equation}
\rho
=
\frac{1}{2}
\left(
\mathbb{I}
+
r_x\sigma_x
+
r_z\sigma_z
\right),
\qquad
r_x^2+r_z^2\leq1.
\label{eq:planar_qubit_state}
\end{equation}
For convenience, we define
\begin{equation}
a=\cos^2\phi,
\quad
b=\sin^2\phi,
\quad
a+b=1.
\label{eq:a_b_definition}
\end{equation}
The four MHQ elements become
\begin{align}
M_{++}
&=
\frac{a}{2}(1+r_z)
+
\frac{\sqrt{ab}}{2}r_x,
\label{eq:Mpp}
\\
M_{+-}
&=
\frac{b}{2}(1+r_z)
-
\frac{\sqrt{ab}}{2}r_x,
\label{eq:Mpm}
\\
M_{-+}
&=
\frac{b}{2}(1-r_z)
+
\frac{\sqrt{ab}}{2}r_x,
\label{eq:Mmp}
\\
M_{--}
&=
\frac{a}{2}(1-r_z)
-
\frac{\sqrt{ab}}{2}r_x.
\label{eq:Mmm}
\end{align}
They satisfy
\begin{equation}
M_{++}+M_{+-}+M_{-+}+M_{--}=1
\end{equation}
and the additional relations
\begin{equation}
M_{++}+M_{--}=a,
\qquad
M_{+-}+M_{-+}=b.
\label{eq:MHQ_pair_sums}
\end{equation}

To determine the largest value of \(L_4\) compatible with
MHQ positivity, it is useful to make the geometry of the
allowed distributions explicit. Introduce
\begin{equation}
u=M_{++},
\qquad
v=M_{+-}.
\end{equation}
Equation~\eqref{eq:MHQ_pair_sums} then gives
\begin{equation}
\boldsymbol{M}(u,v)
=
\left(
u,\,
v,\,
b-v,\,
a-u
\right).
\label{eq:MHQ_uv_parameterization}
\end{equation}
Positivity of all four MHQ elements is equivalent to
\begin{equation}
0\leq u\leq a,
\qquad
0\leq v\leq b.
\label{eq:uv_rectangle}
\end{equation}
For $0<\phi<\pi/2$, inversion of Eqs.~\eqref{eq:Mpp}--\eqref{eq:Mmm}
gives
\begin{equation}
r_z
=
2(u+v)-1,
\qquad
r_x
=
\frac{2(bu-av)}{\sqrt{ab}}.
\label{eq:bloch_from_uv}
\end{equation}
The endpoint cases $\phi=0,\pi/2$, corresponding to the commuting measurements, are obtained separately, with $k(0) = k(\pi/2) = 0$, where $k(\phi)$ denotes the optimal fourth-moment threshold introduced below in Eq.~\eqref{eq:kphi_definition}.
The Bloch-disk condition becomes
\begin{equation}
r_x^2+r_z^2\leq1
\quad\Longleftrightarrow\quad
\frac{u^2}{a}
+
\frac{v^2}{b}
\leq
u+v.
\label{eq:bloch_uv_condition}
\end{equation}
For \(0\leq u\leq a\) and \(0\leq v\leq b\),
\[
\frac{u^2}{a}\leq u,
\qquad
\frac{v^2}{b}\leq v,
\]
so Eq.~\eqref{eq:bloch_uv_condition} is automatically
satisfied. Thus, for fixed \(\phi\), the set of all
MHQ-positive qubit distributions is represented exactly by
the rectangle in Eq.~\eqref{eq:uv_rectangle}.

\subsection{A universal bound on qubit MHQ negativity}
\label{subsec:universal_qubit_bound}

Before specializing to the fourth-moment threshold, we record a bound
on the exact logarithmic MHQ negativity $L_{A,B}(\rho)$ itself, valid
for \emph{every} qubit state and \emph{every} pair of dichotomic
observables. This sharpens the generic bound $L\leq\ln d=\ln2$ of
Proposition~\ref{prop:log_MHQ_properties}.

\begin{proposition}[Universal qubit bound]
\label{prop:universal_qubit_bound}
For every qubit state $\rho$ and every pair of dichotomic
observables $A=\hat{\boldsymbol a}\cdot\boldsymbol\sigma$,
$B=\hat{\boldsymbol b}\cdot\boldsymbol\sigma$,
\begin{equation}
L_{A,B}(\rho)\leq\ln\frac54.
\label{eq:universal_qubit_bound}
\end{equation}
The bound is tight. Equality holds, up to unitary covariance and outcome relabeling, when $\hat{\boldsymbol a}\cdot\hat{\boldsymbol
b}=-\tfrac12$, realized by the pure state with Bloch vector
$\boldsymbol r=(\hat{\boldsymbol a}+\hat{\boldsymbol
b})/|\hat{\boldsymbol a}+\hat{\boldsymbol b}|$; the case $\hat{\boldsymbol a}\cdot\hat{\boldsymbol
b}=+\tfrac12$, with Bloch vector $\boldsymbol r=(\hat{\boldsymbol a}-\hat{\boldsymbol
b})/|\hat{\boldsymbol a}-\hat{\boldsymbol b}|$ gives an equivalent maximizer under outcome relabeling.
\end{proposition}

\begin{proof}
From Eq.~\eqref{eq:general_qubit_MHQ}, we have
\begin{equation}
M_{st}=\frac14\Big[1+st\,c+\boldsymbol r\cdot(s\hat{\boldsymbol
a}+t\hat{\boldsymbol b})\Big],\qquad c=\hat{\boldsymbol
a}\cdot\hat{\boldsymbol b}.
\end{equation}
Introducing $A_\pm=1\pm c\geq0$, so that $A_++A_-=2$, together with
\begin{equation}
u=\boldsymbol r\cdot(\hat{\boldsymbol a}+\hat{\boldsymbol b}),
\qquad
v=\boldsymbol r\cdot(\hat{\boldsymbol a}-\hat{\boldsymbol b}),
\end{equation}
then $M_{++}=(A_++u)/4$, $M_{--}=(A_+-u)/4$,
$M_{+-}=(A_-+v)/4$, $M_{-+}=(A_--v)/4$. Using the elementary
identity $|x+y|+|x-y|=2\max(x,|y|)$ for $x\geq0$,
\begin{equation}
\sum_{s,t}|M_{st}|
=
\frac12\Big[\max(A_+,|u|)+\max(A_-,|v|)\Big].
\label{eq:pairing_identity}
\end{equation}
Since $(\hat{\boldsymbol a}+\hat{\boldsymbol
b})\perp(\hat{\boldsymbol a}-\hat{\boldsymbol b})$ (their dot
product is $|\hat{\boldsymbol a}|^2-|\hat{\boldsymbol b}|^2=0$), and
$|\hat{\boldsymbol a}\pm\hat{\boldsymbol b}|^2=2A_\pm$, the Bloch
constraint $|\boldsymbol r|\leq1$ projected onto this orthogonal
pair gives
\begin{equation}
\frac{u^2}{2A_+}+\frac{v^2}{2A_-}\leq1.
\label{eq:ellipse_constraint}
\end{equation}
With $x=|u|$, $y=|v|$, $A=A_+$, and $B=A_-$, we have $A+B=2$ and
$x^2/(2A)+y^2/(2B)\leq1$. We now bound $\max(A,x)+\max(B,y)$ in the four
regions determined by whether $x\gtrless A$ and $y\gtrless B$:

\emph{(i)} $x\leq A,\ y\leq B$: the sum equals $A+B=2$.

\emph{(ii)} $x>A,\ y\leq B$: the sum is $x+B$. From
Eq.~\eqref{eq:ellipse_constraint}, $x\leq\sqrt{2A}$, so with
$z=\sqrt{2A}$ and $B=2-A=2-z^2/2$,
\begin{equation}
x+B\leq z+2-\frac{z^2}{2}=\frac52-\frac12(z-1)^2\leq\frac52,
\end{equation}
with equality at $z=1$, i.e.\ $A=1/2$.

\emph{(iii)} $x\leq A,\ y>B$: symmetric to (ii), giving the same
bound $5/2$, attained at $B=1/2$.

\emph{(iv)} $x>A,\ y>B$: the sum is $x+y$. By Cauchy--Schwarz,
\begin{align}
x+y &=\sqrt A\cdot\frac{x}{\sqrt A}+\sqrt B\cdot\frac{y}{\sqrt B} \nonumber \\
&\leq\sqrt{(A+B)\Big(\frac{x^2}{A}+\frac{y^2}{B}\Big)}
\leq 2.
\end{align}

In every case $\max(A,x)+\max(B,y)\leq5/2$, so by
Eq.~\eqref{eq:pairing_identity}, $\sum_{s,t}|M_{st}|\leq5/4$ and
hence $L_{A,B}(\rho)\leq\ln(5/4)$. Equality requires either
$A_+=1/2$ (i.e.\ $c=-1/2$), realized by the pure state with Bloch
vector $\boldsymbol r=(\hat{\boldsymbol a}+\hat{\boldsymbol
b})/|\hat{\boldsymbol a}+\hat{\boldsymbol b}|$, or, symmetrically,
$A_-=1/2$ (i.e.\ $c=+1/2$), realized by the pure state with Bloch
vector $\boldsymbol r=(\hat{\boldsymbol a}-\hat{\boldsymbol
b})/|\hat{\boldsymbol a}-\hat{\boldsymbol b}|$. These two solutions
are related by outcome relabeling.
\end{proof}

For the canonical measurement pair of
Eq.~\eqref{eq:canonical_qubit_measurements}, equality is realized at
$\phi=\pi/6$ (or $\phi=\pi/3$, related by outcome relabeling of $B$)
by the pure state with Bloch vector $\boldsymbol r=\hat{\boldsymbol
a}-\hat{\boldsymbol b}=(-\sqrt3/2,0,1/2)$, i.e.\
\begin{equation}
\rho_*
=
\begin{pmatrix}
3/4 & -\sqrt3/4\\
-\sqrt3/4 & 1/4
\end{pmatrix},
\end{equation}
which is exactly the state used in the numerical illustration of
Sec.~\ref{subsec:shadow_numerics}. 
This resolves the origin
of the reported value $L\simeq0.2231\approx\ln(5/4)$ in Eq.~\eqref{eq:qubit_exact_values}.


\subsection{Fourth moment threshold}
The fourth-order moment quantifier is
\begin{equation}
L_4(\rho;A,B)
=
\frac{3}{2}\ln m_2
-
\frac{1}{2}\ln m_4,
\label{eq:L4_qubit}
\end{equation}
where
\begin{equation}
m_n
=
\sum_{s,t}
M_{st}^{\,n}.
\end{equation}
We define the measurement-dependent positivity threshold
\begin{equation}
k(\phi)
=
\max_{\substack{\rho:\\
M_{st}(\rho;A,B)\geq0}}
L_4(\rho;A,B).
\label{eq:kphi_definition}
\end{equation}
By construction,
\begin{equation}
L_4(\rho;A,B)>k(\phi)
\quad\Longrightarrow\quad
\exists\,s,t:
M_{st}(\rho;A,B)<0.
\label{eq:refined_qubit_witness}
\end{equation}

To perform the optimization, write
\begin{equation}
u=\frac{a}{2}+x,
\qquad
v=\frac{b}{2}+y,
\end{equation}
with
\begin{equation}
-\frac{a}{2}\leq x\leq\frac{a}{2},
\qquad
-\frac{b}{2}\leq y\leq\frac{b}{2}.
\end{equation}
The distribution then takes the symmetric form
\begin{equation}
\boldsymbol{M}(x,y)
=
\left(
\frac{a}{2}+x,\,
\frac{b}{2}+y,\,
\frac{b}{2}-y,\,
\frac{a}{2}-x
\right).
\label{eq:symmetric_MHQ_parameterization}
\end{equation}
Consequently,
\begin{align}
m_2
&=
\frac{a^2+b^2}{2}
+
2(x^2+y^2),
\label{eq:m2_xy}
\\
m_4
&=
\frac{a^4+b^4}{8}
+
3a^2x^2
+
3b^2y^2
+
2(x^4+y^4).
\label{eq:m4_xy}
\end{align}
Hence \(L_4\) depends only on \(x^2\) and \(y^2\). Before proceeding further, we state a lemma that will help us to find exact expressions for $k(\phi).$

\begin{lemma}[Vertex domination]
\label{lemma:vertex_domination}
Let $a,b>0$ with $a+b=1$. On the rectangle $0\le X\le A=a^2/4$,
$0\le Y\le B=b^2/4$, the function
\begin{equation}
F(X,Y)=\frac{m_2(X,Y)^3}{m_4(X,Y)}
\end{equation}
attains its maximum at one of the four corners of the rectangle.
\end{lemma}

The proof is given in Appendix~\ref{app:vertex_domination}. 
By Lemma~\ref{lemma:vertex_domination}, this maximum is attained at one of the four vertices of the rectangle, and the corresponding MHQ distributions reduce, up to permutations, to the following three distinct families:
\begin{align}
\boldsymbol{M}_0
&=
\{a,b,0,0\},
\label{eq:M_family_0}
\\
\boldsymbol{M}_a
&=
\left\{
\frac{a}{2},
\frac{a}{2},
b,
0
\right\},
\label{eq:M_family_a}
\\
\boldsymbol{M}_b
&=
\left\{
a,
\frac{b}{2},
\frac{b}{2},
0
\right\}.
\label{eq:M_family_b}
\end{align}
The central distribution
\(
\{a/2,b/2,b/2,a/2\}
\)
gives the same value of \(L_4\) as
\(\boldsymbol{M}_0\); this follows because splitting each
entry into two equal parts rescales \(m_2\) and \(m_4\) in
such a way that the combination in
Eq.~\eqref{eq:L4_qubit} remains unchanged.

Evaluating \(L_4\) on the three families gives
\begin{align}
K_0(\phi)
&=
\frac{3}{2}
\ln\!\left(a^2+b^2\right)
-
\frac{1}{2}
\ln\!\left(a^4+b^4\right),
\label{eq:K0}
\\
K_a(\phi)
&=
\frac{3}{2}
\ln\!\left(
\frac{a^2}{2}+b^2
\right)
-
\frac{1}{2}
\ln\!\left(
\frac{a^4}{8}+b^4
\right),
\label{eq:Ka}
\\
K_b(\phi)
&=
\frac{3}{2}
\ln\!\left(
a^2+\frac{b^2}{2}
\right)
-
\frac{1}{2}
\ln\!\left(
a^4+\frac{b^4}{8}
\right).
\label{eq:Kb}
\end{align}

We can therefore state the optimized qubit criterion as
follows.

\begin{proposition}
\label{prop:optimal_qubit_L4_threshold}
For a qubit state \(\rho\) and the pair of dichotomic
observables in Eq.~\eqref{eq:canonical_qubit_measurements},
the largest value of \(L_4\) compatible with an everywhere
nonnegative MHQ distribution is
\begin{equation}
k(\phi)
=
\max
\left\{
K_a(\phi),
K_0(\phi),
K_b(\phi)
\right\}.
\label{eq:kphi_max}
\end{equation}
Consequently,
\begin{equation}
L_4(\rho;A,B)>k(\phi)
\end{equation}
certifies negativity of the corresponding MHQ distribution.
The threshold is optimal among certification rules based
only on the value of \(L_4\): no smaller threshold is valid
for all MHQ-positive qubit distributions at the same
measurement angle \(\phi\).
\end{proposition}

The threshold obeys the symmetry
\begin{equation}
k(\phi)
=
k\!\left(\frac{\pi}{2}-\phi\right),
\end{equation}
with
\begin{equation}
K_b(\phi)
=
K_a\!\left(\frac{\pi}{2}-\phi\right).
\end{equation}
Let \(\phi_c\in(0,\pi/4)\) denote the crossover angle defined
implicitly by
\begin{equation}
K_a(\phi_c)=K_0(\phi_c).
\label{eq:phi_c_definition}
\end{equation}
Numerically,
$\phi_c
\simeq
0.6956.$
The optimized threshold can then be written as
\begin{equation}
k(\phi)
=
\begin{cases}
K_a(\phi),
&
0\leq\phi\leq\phi_c,
\\[1mm]
K_0(\phi),
&
\phi_c\leq\phi
\leq
\dfrac{\pi}{2}-\phi_c,
\\[2mm]
K_b(\phi),
&
\dfrac{\pi}{2}-\phi_c
\leq\phi\leq\dfrac{\pi}{2}.
\end{cases}
\label{eq:kphi_piecewise}
\end{equation}

The maximizing positive MHQ distribution is not unique in
general. In the central angular region, the value \(K_0\) can
be attained both by boundary distributions of the form
\(\{a,b,0,0\}\) and by the distribution
\(\{a/2,b/2,b/2,a/2\}\). Outside this region, the maximizing
distributions are represented by
Eqs.~\eqref{eq:M_family_a} and \eqref{eq:M_family_b}.
Accordingly, the optimized threshold cannot in general be
obtained by restricting the optimization to pure states.

The criterion in Eq.~\eqref{eq:refined_qubit_witness}
improves upon the universal condition \(L_4>0\) whenever
\(k(\phi)<0\) (see Fig.~\ref{fig:qubit_threshold}). Its optimality, however, is restricted to
witnesses that use only the single scalar quantity \(L_4\);
it does not exclude stronger tests constructed from the
individual moments \(m_2\) and \(m_4\), from additional
moments, or from the full MHQ distribution.

\begin{figure}[t]
\centering
\includegraphics[width=0.9\linewidth]
{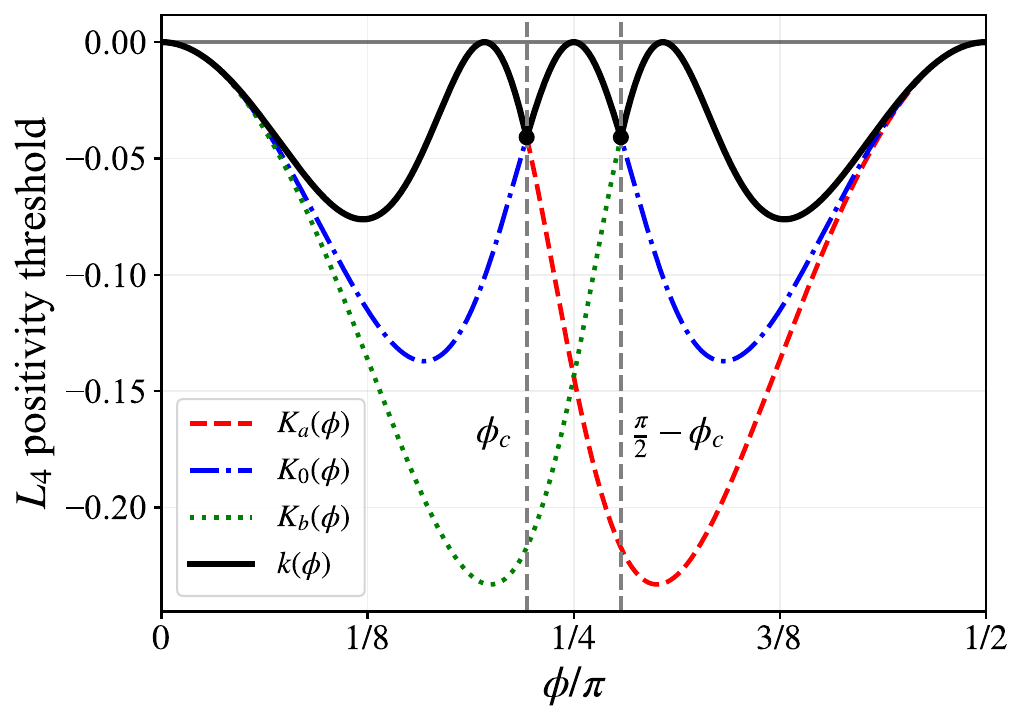}
\caption{
\justifying
\small
The solid black curve $k(\phi)$ denotes the optimal threshold as a function of the angle $\phi \in [0, \pi/2]$ between the Bloch vectors of the observables $A$ and $B$. The dashed red, dash-dotted blue, and dotted green curves represent the branches $K_a(\phi)$, $K_0(\phi)$, and $K_b(\phi)$, respectively. The black points and vertical dashed lines mark the transition points $\phi = \phi_c$ and $\phi = \pi/2 - \phi_c$, at which the optimal strategy switches from the $K_a$ to the $K_0$ branch, and from the $K_0$ to the $K_b$ branch, respectively.
}
\label{fig:qubit_threshold}
\end{figure}

\section{Extension to continuous-variable systems}
\label{s6}

The moment construction developed above is not restricted to
finite-dimensional quasiprobability distributions. It also
applies to continuous MH distributions, provided
the relevant absolute moments are finite. We illustrate this
extension for a single bosonic mode and the canonical
quadratures:
$\hat{x}
=
{(\hat{a}+\hat{a}^{\dagger})}/{\sqrt{2}},
\;
\hat{p}
=
{(\hat{a}-\hat{a}^{\dagger})}/{\sqrt{2} i},
\, \text{satisfying}\;
[\hat{x},\hat{p}]=i.$
Their generalized eigenstates satisfy
\begin{equation}
\langle x|p\rangle
=
\frac{1}{\sqrt{2\pi}}e^{ipx}.
\label{eq:xp_overlap}
\end{equation}

For a single-mode state $\rho$, the MHQ distribution
associated with the position and momentum measurements is
defined by
\begin{align}
M_{\rho}(x,p)
&=
\frac{1}{2}
\operatorname{Tr}
\left[
\left(
|x\rangle\langle x|
|p\rangle\langle p|
+
|p\rangle\langle p|
|x\rangle\langle x|
\right)
\rho
\right]
\nonumber\\
&=
\operatorname{Re}
\left[
\langle p|x\rangle
\langle x|\rho|p\rangle
\right]
\nonumber\\
&=
\frac{1}{\sqrt{2\pi}}
\operatorname{Re}
\left[
e^{-ipx}\langle x|\rho|p\rangle
\right].
\label{eq:CV_MHQ_definition}
\end{align}

The distribution is real and normalized,
\begin{equation}
\int_{\mathbb{R}^{2}}
M_{\rho}(x,p)\,dx\,dp
=
1,
\label{eq:CV_MHQ_normalization}
\end{equation}
and its marginals reproduce the position and momentum
probability densities:
\begin{align}
\int_{\mathbb{R}}
M_{\rho}(x,p)\,dp
=
\langle x|\rho|x\rangle, \quad
\int_{\mathbb{R}}
M_{\rho}(x,p)\,dx
=
\langle p|\rho|p\rangle.
\end{align}
As in the discrete case, noncommutativity is necessary but
not sufficient for negativity: the sign of
\(M_{\rho}(x,p)\) depends on both the state and the chosen
quadratures. Whenever the corresponding integrals exist, we define the regular and absolute moments as,
\begin{align}
m_n(\rho)
&=
\int_{\mathbb{R}^{2}}
[M_{\rho}(x,p)]^n\,dx\,dp,
\label{eq:CV_MHQ_moments}\\
\widetilde m_n(\rho)
&=
\int_{\mathbb{R}^{2}}
|M_{\rho}(x,p)|^n\,dx\,dp.
\label{eq:CV_absolute_moments}
\end{align}
The logarithmic MHQ negativity and its moment-based lower
bounds retain the same form as in the discrete setting:
\begin{equation}
L(\rho)
=
\ln\widetilde m_1(\rho),
\label{eq:CV_log_MHQ}
\end{equation}
and, for \(r>2\),
\begin{equation}
L_r(\rho)
=
\frac{1}{2-r}
\left[
\ln\widetilde m_r(\rho)
+
(1-r)\ln\widetilde m_2(\rho)
\right]
\leq
L(\rho).
\label{eq:CV_Lr}
\end{equation}
Since \(M_{\rho}(x,p)\) is real,
$\widetilde m_{2n}(\rho)=m_{2n}(\rho),$ 
and in particular
\begin{equation*}
L_4(\rho)
=
\frac{3}{2}\ln m_2(\rho)
-
\frac{1}{2}\ln m_4(\rho).
\end{equation*}

We next evaluate these quantities for two representative
families of non-Gaussian states.

\subsection{Fock states}
\label{subsec:CV_Fock}

For the Fock state $|m\rangle$, the position- and
momentum-space wavefunctions are
\begin{align}
\langle x|m\rangle
&=
\frac{
H_m(x)e^{-x^2/2}
}{
\pi^{1/4}\sqrt{2^m m!}
},\\
\langle p|m\rangle
&=
\frac{
(-i)^mH_m(p)e^{-p^2/2}
}{
\pi^{1/4}\sqrt{2^m m!}
}.
\end{align}
For the pure state $\rho_m=|m\rangle\langle m|$,
substitution into the definition of the MHQ distribution yields
\begin{equation}
M_m(x,p)
=
\frac{
H_m(x)H_m(p)
}{
\sqrt{2}\,\pi\,2^m m!
}
\exp\left[
-\frac{x^2+p^2}{2}
\right]
\cos\left(
px-\frac{m\pi}{2}
\right).
\label{eq:Fock_MHQ}
\end{equation}
The phase shift $m\pi/2$ originates from the Fourier phase
of the Hermite functions. Consequently, even and odd Fock
states exhibit cosine- and sine-type oscillatory structures,
respectively.

\subsection{Photon-added coherent states}
\label{subsec:CV_PACS}
An $m$-photon-added coherent state is defined as
\begin{equation}
|\alpha,m\rangle
=
\frac{
(\hat a^\dagger)^m|\alpha\rangle
}{
\sqrt{
m!L_m(-|\alpha|^2)
}
},
\end{equation}
where $L_m$ denotes the Laguerre polynomial. We define
\begin{equation}
x_0=\sqrt{2}\operatorname{Re}\alpha,
\qquad
p_0=\sqrt{2}\operatorname{Im}\alpha.
\end{equation}

Up to an overall global phase, the coherent-state wavefunctions
may be written as
\begin{align}
\langle x|\alpha\rangle
&=
\pi^{-1/4}
\exp\left[
-\frac{(x-x_0)^2}{2}
+
ip_0x
-
\frac{i}{2}x_0p_0
\right],\\
\langle p|\alpha\rangle
&=
\pi^{-1/4}
\exp\left[
-\frac{(p-p_0)^2}{2}
-
ix_0p
+
\frac{i}{2}x_0p_0
\right].
\end{align}

Using
\(
\hat a^\dagger
=
\frac{1}{\sqrt2}
\left(
x-\frac{d}{dx}
\right),
\)
one finds
\begin{equation}
(\hat a^\dagger)^m
\langle x|\alpha\rangle
=
2^{-m/2}
H_m
\left(
x-\frac{\alpha}{\sqrt2}
\right)
\langle x|\alpha\rangle.
\end{equation}
Hence,
\begin{equation}
\langle x|\alpha,m\rangle
=
\frac{
H_m\!\left(x-\alpha/\sqrt2\right)
}{
\sqrt{
2^m m!L_m(-|\alpha|^2)
}
}
\langle x|\alpha\rangle.
\end{equation}
Similarly,
\begin{equation}
\langle p|\alpha,m\rangle
=
\frac{
(-i)^m
H_m\!\left(p+i\alpha/\sqrt2\right)
}{
\sqrt{
2^m m!L_m(-|\alpha|^2)
}
}
\langle p|\alpha\rangle.
\end{equation}

Substitution into the definition of the MHQ distribution gives
\begin{align}
M_{\alpha,m}(x,p)
&=
\frac{
e^{-\frac12[(x-x_0)^2+(p-p_0)^2]}
}{
\sqrt2\,\pi\,
2^m m!L_m(-|\alpha|^2)
}
\nonumber\\
&\quad\times
\operatorname{Re}
\Bigg[
i^m
e^{i\Theta_\alpha(x,p)}
H_m
\left(
x-\frac{\alpha}{\sqrt2}
\right)
H_m
\left(
p-\frac{i\alpha^*}{\sqrt2}
\right)
\Bigg],
\label{eq:PACS_MHQ}
\end{align}
where $\Theta_\alpha(x,p) = p_0x+x_0p-x_0p_0-px.$ Unlike the expression for Fock states, the polynomial factors
are generally complex when \(\alpha\notin\mathbb{R}\).
Equation~\eqref{eq:PACS_MHQ} reduces to
Eq.~\eqref{eq:Fock_MHQ} when \(\alpha=0\), as required.
The Gaussian envelope in Eq.~\eqref{eq:PACS_MHQ} is
displaced by \((x_0,p_0)\), while photon addition introduces
higher-order polynomial and interference structures. These
features can modify both the magnitude and spatial
distribution of MHQ negativity. We evaluate the corresponding
absolute moments numerically rather than assuming monotonic
behavior with either \(m\) or \(|\alpha|\).

The numerical results for the parameter ranges considered are
shown in Fig.~\ref{fig:cv_mh}. For both state families, the
fourth-moment quantity satisfies
\[
L_4\geq L_6,
\]
in agreement with the general hierarchy derived above. The
dependence of $L$, $L_4$, and $L_6$ on the excitation number
is obtained numerically and is not intended to establish a
general monotonicity property. Rather, these examples illustrate
that low-order moments can capture a substantial part of the
MHQ negativity without reconstructing the full
quasiprobability distribution.

\begin{figure}[t]
\centering
\includegraphics[width=1\linewidth]{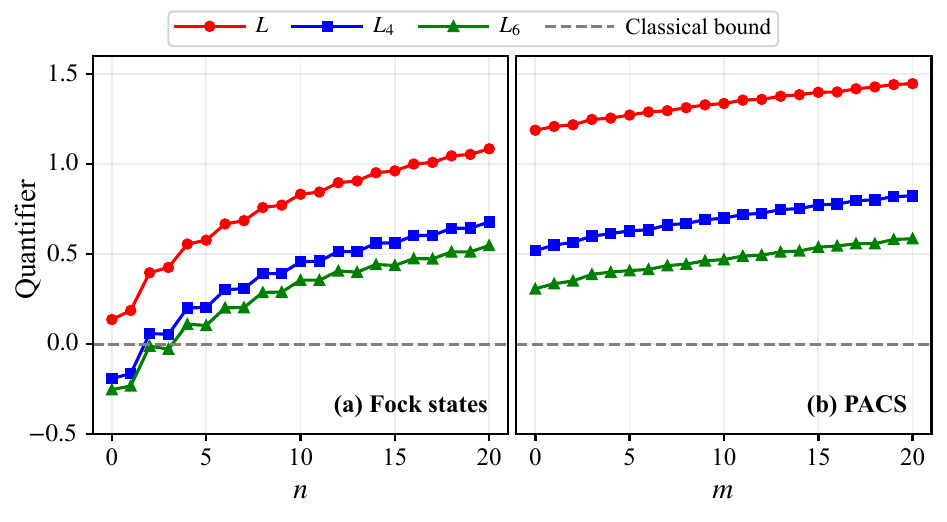}
\caption{
\justifying
\small
Comparison of the logarithmic MHQ negativity $L$ and the
moment-based lower bounds $L_4$ and $L_6$ for
(a) Fock states and (b) photon-added coherent states (for $\alpha = 2$) for the
parameter ranges considered. The inequalities
$
L\geq L_4\geq L_6
$
hold throughout the displayed ranges, illustrating that the
fourth-moment bound provides a tighter estimate than the
sixth-moment bound. The observed dependence on the excitation
number is numerical and is not asserted as a general
monotonicity property.
}
\label{fig:cv_mh}
\end{figure}

\section{Multicopy representation and estimation of MHQ moments}
\label{s7}

The lower bounds introduced in Sec.~\ref{s4} depend on the
even moments of the MHQ distribution.
A direct evaluation of these moments from the individual
quasiprobability elements would require reconstructing the
MHQ distribution. We now show that this reconstruction is not
necessary: every ordinary MHQ moment admits an exact
multicopy representation. We then discuss two possible routes
for estimating the corresponding expectation values.

\subsection{Multicopy representation}
\label{MHQ_operator}

For the projective observables
$
A=\sum_i a_i\Pi_i^a,
\;
B=\sum_j b_j\Pi_j^b,
$, we define the Hermitian operators
\begin{equation}
Y_{ij}
=
\frac{1}{2}
\left(
\Pi_i^a\Pi_j^b
+
\Pi_j^b\Pi_i^a
\right).
\label{eq:Yij_definition}
\end{equation}
The MHQ elements can then be written as
\begin{equation}
M_{ij}(\rho;A,B)
=
\operatorname{Tr}(Y_{ij}\rho).
\label{eq:MHQ_Yij}
\end{equation}
For every positive integer \(n\), the identity
\begin{equation}
\left[
\operatorname{Tr}(X\rho)
\right]^n
=
\operatorname{Tr}
\left[
X^{\otimes n}\rho^{\otimes n}
\right]
\label{eq:multicopy_identity}
\end{equation}
gives
\begin{align}
m_n(\rho;A,B)
&=
\sum_{i,j}
\left[
M_{ij}(\rho;A,B)
\right]^n
\nonumber\\
&=
\operatorname{Tr}
\left[
\mathcal M_n^{A,B}\rho^{\otimes n}
\right],
\label{eq:MHQ_multicopy_representation}
\end{align}
where
\begin{equation}
\mathcal M_n^{A,B}
=
\sum_{i,j}
Y_{ij}^{\otimes n}.
\label{eq:Mn_definition}
\end{equation}
Since every \(Y_{ij}\) is Hermitian,
\begin{equation}
\left(\mathcal M_n^{A,B}\right)^\dagger
=
\mathcal M_n^{A,B}.
\end{equation}
Thus \(m_n\) is the expectation value of a Hermitian
observable acting on \(n\) identical copies of the state. The
operator \(\mathcal M_n^{A,B}\) is also invariant under
permutations of the copies, reflecting the symmetry of the
polynomial \(m_n\).

In particular, we consider the second and fourth MHQ moments, which we require to construct the quantifier $L_{4;A,B}$,
\begin{align}
m_2
&=
\operatorname{Tr}
\left[
\mathcal M_2^{A,B}\rho^{\otimes2}
\right],
&
\mathcal M_2^{A,B}
&=
\sum_{i,j}Y_{ij}\otimes Y_{ij},
\label{eq:m2_multicopy}
\\
m_4
&=
\operatorname{Tr}
\left[
\mathcal M_4^{A,B}\rho^{\otimes4}
\right],
&
\mathcal M_4^{A,B}
&=
\sum_{i,j}Y_{ij}^{\otimes4}.
\label{eq:m4_multicopy}
\end{align}
No expansion of these operators into all projector orderings is
required for the subsequent analysis. The compact forms in
Eqs.~\eqref{eq:m2_multicopy} and
\eqref{eq:m4_multicopy} already specify the collective
observables whose expectation values yield the desired
moments. These multicopy identities establish experimental
accessibility in principle, which we discuss in the next few subsections.

\subsection{Phase-sensitive interferometric estimation}
\label{MHQ_interferometer}
A physically intuitive route for accessing the MHQ moments is through interferometry. In particular, ancilla-assisted Mach–Zehnder interferometry provides access to the complex characteristic function $\chi_n (\lambda)$, whose phase-sensitive response near $\lambda=0$ yields the desired moment, as detailed below. This allows one to relate the MHQ moments directly to experimentally measurable quantities.

A Mach--Zehnder interferometer consists of two spatial paths, which define a two-dimensional Hilbert space spanned by $\{\ket{0},\ket{1}\}$. This degree of freedom is referred to as the \emph{path degree of freedom}. The quantum state $\rho$, whose properties are to be probed, is encoded in an auxiliary system, referred to as the \emph{internal degree of freedom}. The total system is therefore described by the tensor product space of path and internal degrees of freedom.

The interferometer typically contains two $50$-$50$ beam splitters (BS-I, BS-II), a phase shifter ($ PS (\theta) $), mirrors (M-I, M-II), and a controlled unitary $(U_n (\lambda))$ acting on the internal system (see Fig.~\ref{fig:MHQ_interferometer}). The corresponding unitary operations are as follows:
\begin{itemize}
\item The beam splitter (BS) is represented by the Hadamard operation $H$ acting on the path degree of freedom.
\item A phase shift (PS) of $\theta$ in the path $\ket{0}$ is implemented by the unitary $e^{i\theta}\ket{0}\bra{0} + \ket{1}\bra{1}$.
\item A controlled unitary $U$ acting on the internal system is given by $\ket{0}\bra{0}\otimes \mathbb{I} + \ket{1}\bra{1}\otimes U$.
\item Mirrors are represented by the Pauli-$X$ operation on the path space.
\end{itemize}
A standard ancilla-assisted interferometer can estimate the
complex quantity
\begin{equation}
\chi_n(\lambda)
=
\operatorname{Tr}
\left[
U_n(\lambda)\rho^{\otimes n}
\right],
\quad \text{where}\;\;
U_n(\lambda)
=
e^{i\lambda\mathcal M_n^{A,B}},
\label{eq:characteristic_function}
\end{equation}
provided the controlled unitary \(U_n(\lambda)\) can be
implemented.

To see how the interferometric signal is related to
\(\chi_n(\lambda)\), let the path degree of freedom act as an
ancilla controlling \(U_n(\lambda)\). Introducing a relative
phase \(\theta\) between the two paths gives the output
probability
\begin{equation}
P_0(\theta,\lambda)
=
\frac{1}{2}
\left[
1+
\operatorname{Re}
\left(
e^{-i\theta}\chi_n(\lambda)
\right)
\right].
\label{eq:interferometric_probability}
\end{equation}
Consequently,
\begin{align}
\operatorname{Re}\chi_n(\lambda)
&=
2P_0(0,\lambda)-1,
\label{eq:real_chi}
\\
\operatorname{Im}\chi_n(\lambda)
&=
2P_0\!\left(\frac{\pi}{2},\lambda\right)-1.
\label{eq:imag_chi}
\end{align}
The modulus
\(
|\chi_n(\lambda)|
\)
determines the fringe visibility, whereas the shift of the
fringe determines its phase. Both quantities, or equivalently
the two probabilities in
Eqs.~\eqref{eq:real_chi} and \eqref{eq:imag_chi}, are required
to reconstruct the full complex signal. Expanding Eq.~\eqref{eq:characteristic_function} around
\(\lambda \rightarrow 0\) gives
\begin{align}
\chi_n(\lambda)
&=
1
+
i\lambda
\operatorname{Tr}
\left[
\mathcal M_n^{A,B}\rho^{\otimes n}
\right]
-\frac{\lambda^2}{2}
\operatorname{Tr}
\left[
\left(\mathcal M_n^{A,B}\right)^2
\rho^{\otimes n}
\right]
+
O(\lambda^3)
\nonumber\\
&=
1+i\lambda m_n+O(\lambda^2).
\label{eq:chi_expansion}
\end{align}
The moment is therefore encoded in the derivative of the
imaginary part:
\begin{equation}
m_n
=
\left.
\frac{d}{d\lambda}
\operatorname{Im}\chi_n(\lambda)
\right|_{\lambda \rightarrow 0}.
\label{eq:moment_derivative}
\end{equation}
For a small but finite value of \(\lambda\),
\begin{equation}
\frac{\operatorname{Im}\chi_n(\lambda)}{\lambda}
=
m_n
-
\frac{\lambda^2}{6}
\operatorname{Tr}
\left[
\left(\mathcal M_n^{A,B}\right)^3
\rho^{\otimes n}
\right]
+
O(\lambda^4).
\label{eq:finite_lambda_estimator}
\end{equation}
Alternatively, a symmetric finite difference may also be used:
\begin{equation}
m_n
=
\lim_{\lambda\rightarrow0}
\frac{
\operatorname{Im}\chi_n(\lambda)
-
\operatorname{Im}\chi_n(-\lambda)
}{
2\lambda
}.
\label{eq:symmetric_derivative}
\end{equation}

For the two moments required by \(L_{4;A,B}\), one would use
$U_2(\lambda)
=
e^{i\lambda\mathcal M_2^{A,B}}$
on two copies and
$U_4(\lambda)
=
e^{i\lambda\mathcal M_4^{A,B}}$
on four copies. This provides a formally valid
interferometric procedure, but its practical implementation
requires controlled evolutions generated by the generally
nonlocal multicopy observables
\(\mathcal M_2^{A,B}\) and \(\mathcal M_4^{A,B}\). The
construction should therefore be understood as an
ancilla-assisted measurement prescription rather than an
immediate laboratory implementation for arbitrary systems.

\begin{figure}[t]
\centering
\includegraphics[width=0.9\linewidth]
{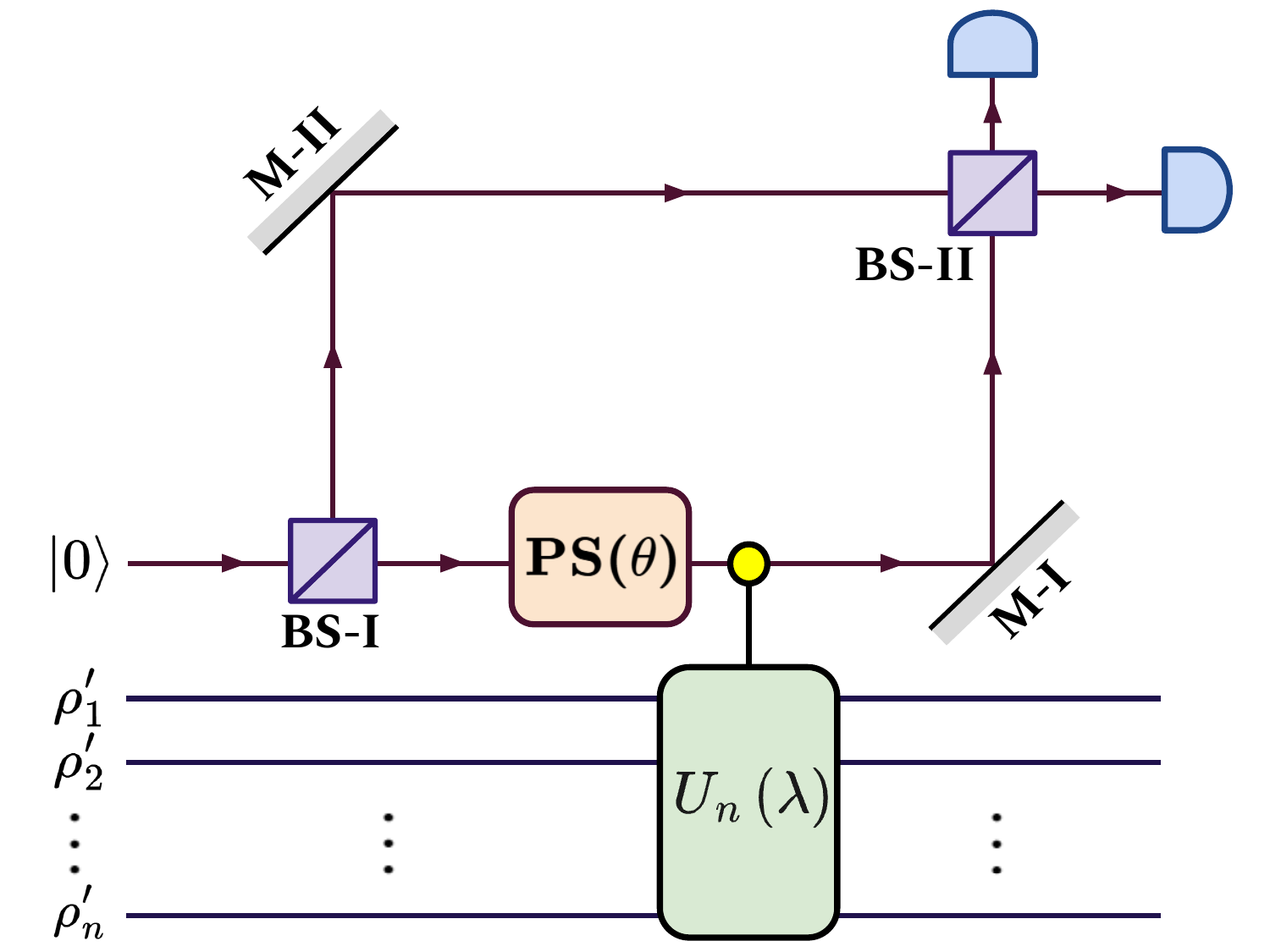}
\caption{
\justifying
\small
Ancilla-assisted interferometric estimation of the
characteristic function
\(
\chi_n(\lambda)
=
\operatorname{Tr}
[e^{i\lambda\mathcal M_n^{A,B}}\rho^{\otimes n}]
\).
The path degree of freedom controls the unitary
\(e^{i\lambda\mathcal M_n^{A,B}}\) acting on \(n\) copies of
the state. Measurements at two relative phases,
\(\theta=0\) and \(\theta=\pi/2\), yield the real and
imaginary parts of \(\chi_n(\lambda)\), respectively. The
MHQ moment \(m_n\) is obtained from the derivative of the
imaginary part at \(\lambda \rightarrow 0\), rather than from the
visibility alone.
}
\label{fig:MHQ_interferometer}
\end{figure}

\subsection{Estimation from classical shadows}
\label{subsec:MHQ_shadows}

Collective implementation of
\(\mathcal M_n^{A,B}\) may become difficult as the moment
order increases. An alternative is to estimate the multicopy
expectation values from randomized single-copy measurements.
Classical-shadow protocols associate each measurement outcome
with a random matrix \(\widehat\rho\) satisfying
\begin{equation}
\mathbb E[\widehat\rho]=\rho.
\label{eq:shadow_unbiasedness}
\end{equation}
Independent snapshots can therefore be combined to estimate
polynomial functions of the state. Let
\begin{equation}
\mathcal S_N
=
\{
\widehat\rho_1,\ldots,\widehat\rho_N
\}
\end{equation}
be \(N\) independent shadow snapshots. For distinct indices
\(i_1,\ldots,i_n\),
\begin{align}
&\mathbb E
\left[
\operatorname{Tr}
\left(
\mathcal M_n^{A,B}
\widehat\rho_{i_1}
\otimes\cdots\otimes
\widehat\rho_{i_n}
\right)
\right]
\nonumber\\
&
=
\operatorname{Tr}
\left[
\mathcal M_n^{A,B}\rho^{\otimes n}
\right] \nonumber\\
&= m_n.
\label{eq:shadow_kernel_unbiased}
\end{align}
This gives the unbiased U-statistic estimator~\cite{huang2020predicting}
\begin{equation}
\widehat m_n
=
\binom{N}{n}^{-1}
\sum_{1\leq i_1<\cdots<i_n\leq N}
\operatorname{Tr}
\left[
\mathcal M_n^{A,B}
\bigotimes_{\ell=1}^{n}
\widehat\rho_{i_\ell}
\right].
\label{eq:Mn_U_statistic}
\end{equation}
Because \(\mathcal M_n^{A,B}\) is permutation invariant, no
additional symmetrization over the selected snapshots is
required. When evaluation of all
\(\binom{N}{n}\) tuples is computationally costly, an
incomplete U-statistic obtained by randomly sampling distinct
tuples remains unbiased with respect to the tuple sampling.

For the second and fourth moments,
\begin{align}
\widehat m_2
&=
\binom{N}{2}^{-1}
\sum_{i<j}
\operatorname{Tr}
\left[
\mathcal M_2^{A,B}
\widehat\rho_i\otimes\widehat\rho_j
\right],
\label{eq:m2_shadow_estimator}
\\
\widehat m_4
&=
\binom{N}{4}^{-1}
\sum_{i<j<k<\ell}
\operatorname{Tr}
\left[
\mathcal M_4^{A,B}
\widehat\rho_i\otimes\widehat\rho_j
\otimes\widehat\rho_k\otimes\widehat\rho_\ell
\right].
\label{eq:m4_shadow_estimator}
\end{align}
These estimators access the required moments without reconstructing the individual MHQ elements.

The corresponding plug-in estimator is
\begin{equation}
\widehat L_4
=
\frac{3}{2}\ln\widehat m_2
-
\frac{1}{2}\ln\widehat m_4.
\label{eq:L4_plugin_estimator}
\end{equation}
Although
\(\widehat m_2\) and \(\widehat m_4\) are unbiased,
\(\widehat L_4\) is generally biased because of the nonlinear
logarithms.
For a statistically valid certification statement, one may
construct simultaneous confidence bounds
\begin{equation}
m_2\geq m_2^{\mathrm{LB}},
\qquad
m_4\leq m_4^{\mathrm{UB}},
\end{equation}
with a prescribed confidence level. Since \(L_4\) increases
with \(m_2\) and decreases with \(m_4\), the quantity
\begin{equation}
L_4^{\mathrm{LB}}
=
\frac{3}{2}
\ln m_2^{\mathrm{LB}}
-
\frac{1}{2}
\ln m_4^{\mathrm{UB}}
\label{eq:L4_confidence_bound}
\end{equation}
is then a conservative lower confidence bound on the exact
\(L_4\).

\subsection{Single-qubit numerical illustration}
\label{subsec:shadow_numerics}

\begin{figure*}[t]
\centering

\begin{subfigure}{0.32\textwidth}
    \centering
    \includegraphics[width=\textwidth]{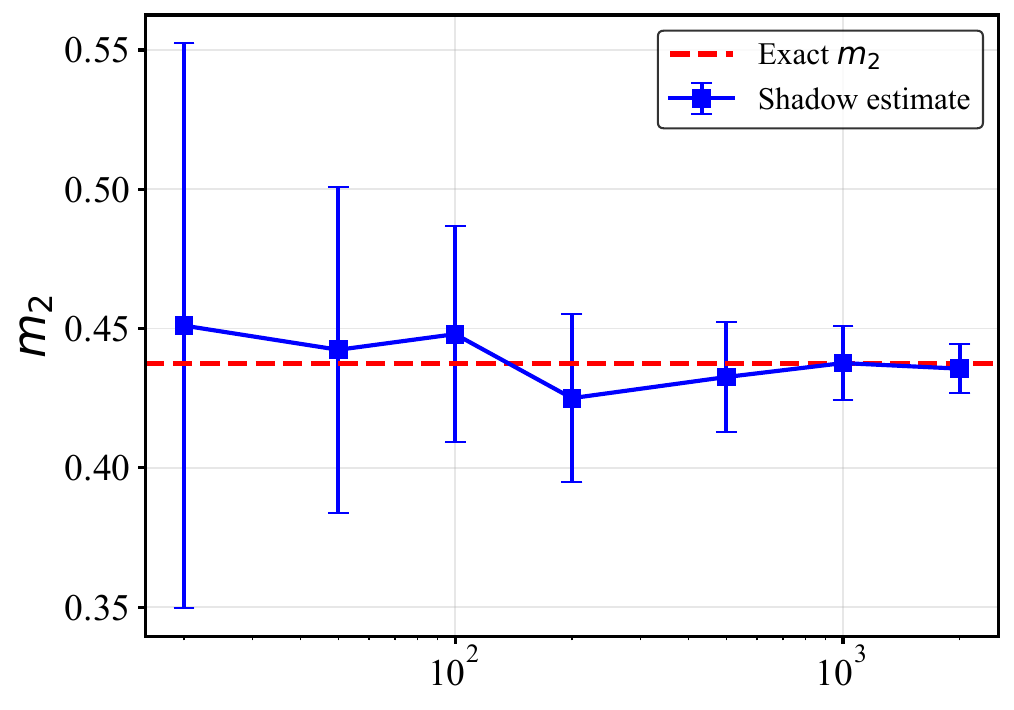}
    \caption{}
    \label{fig:m2_shadow}
\end{subfigure}
\hfill
\begin{subfigure}{0.32\textwidth}
    \centering
    \includegraphics[width=\textwidth]{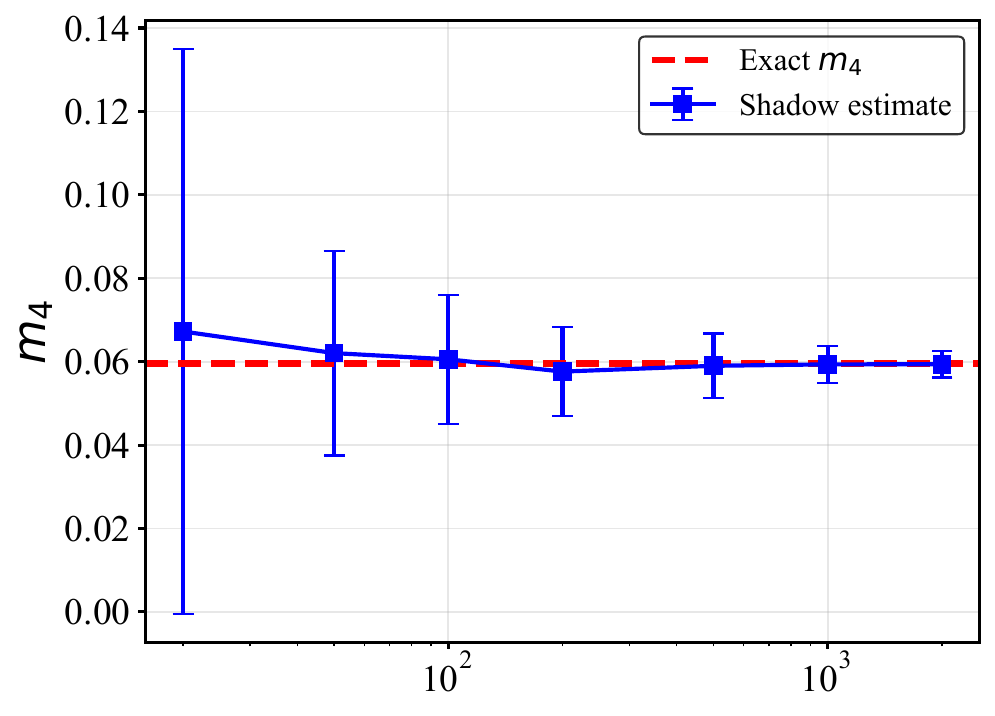}
    \caption{}
    \label{fig:m4_shadow}
\end{subfigure}
\hfill
\begin{subfigure}{0.32\textwidth}
    \centering
    \includegraphics[width=\textwidth]{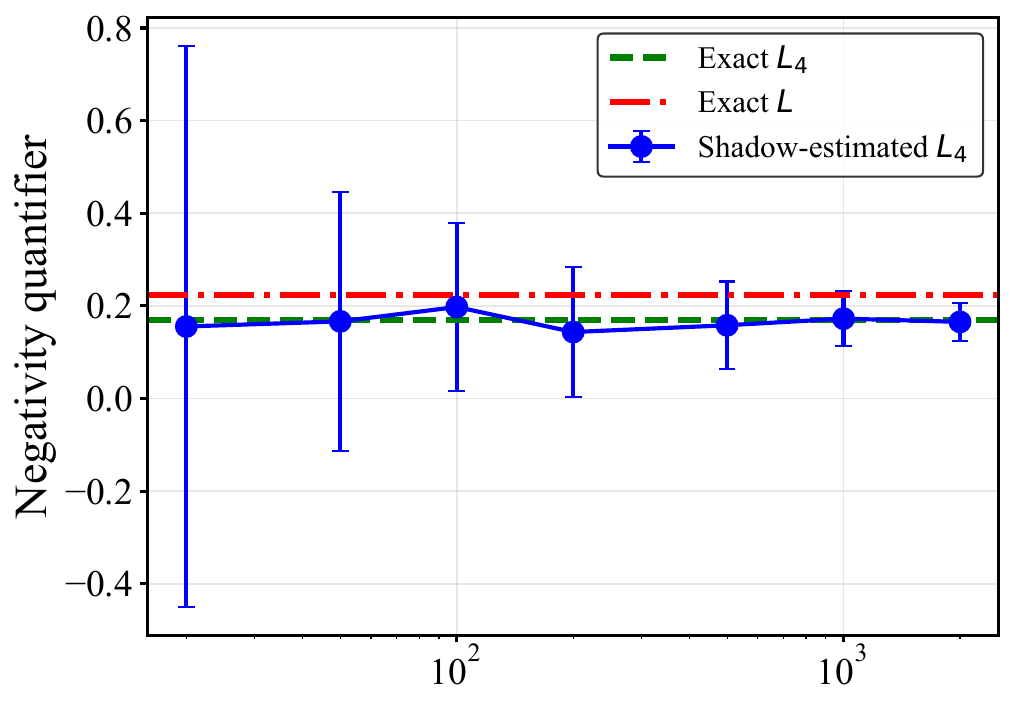}
    \caption{}
    \label{fig:L4_shadow}
\end{subfigure}

\caption{
\justifying
\small
Nonlinear classical-shadow estimation for the qubit state and
measurement pair in Eq.~\eqref{eq:shadow_qubit_observables}.
Panels (a) and (b) show the U-statistic estimates of the
second and fourth MHQ moments, respectively. Panel (c) shows
the plug-in estimate
\(
\widehat L_4
=
\frac32\ln\widehat m_2-\frac12\ln\widehat m_4
\).
Horizontal reference lines denote the exact values.
Points denote the mean over $100$ independent realizations and the caps denote the corresponding standard deviation.
The finite-sample estimator
\(\widehat L_4\) is not constrained to remain below the exact
logarithmic MHQ negativity in every realization; the
inequality \(L_4\leq L\) concerns the underlying exact
quantities.
}
\label{fig:shadow_mhq}
\end{figure*}

We illustrate the estimator for the qubit observables
\begin{equation}
A=\sigma_z,
\quad
B=
\cos(2\alpha)\sigma_z
+
\sin(2\alpha)\sigma_x,
\;\; \text{with}\;\;
\alpha=\frac{\pi}{6}.
\label{eq:shadow_qubit_observables}
\end{equation}
We use the state that globally maximizes $L_{A,B}$ for this
observable pair (Proposition~\ref{prop:universal_qubit_bound}),
with exact Bloch vector $\boldsymbol r=\hat{\boldsymbol
a}-\hat{\boldsymbol b}=(-\sqrt3/2,0,1/2)$:
\begin{equation}
\rho
=
\begin{pmatrix}
3/4 & -\sqrt3/4\\
-\sqrt3/4 & 1/4
\end{pmatrix}
=
\begin{pmatrix}
0.75 & -0.4330\\
-0.4330 & 0.25
\end{pmatrix}.
\label{eq:shadow_qubit_state}
\end{equation}
The exact values are
\begin{align}
m_2&=7/16=0.4375, \quad m_4=61/1024\simeq0.05957 \nonumber \\
L&=\ln(5/4)\simeq0.223144,
\quad L_4 \simeq 0.170281. \label{eq:qubit_exact_values}
\end{align}
These numbers should be interpreted as properties of this
specific state and measurement pair, not as general bounds
for qubit systems.

Randomized Pauli measurements are used to generate
single-qubit shadow snapshots. The moments are estimated using
Eqs.~\eqref{eq:m2_shadow_estimator} and
\eqref{eq:m4_shadow_estimator}, or incomplete versions of
these U-statistics when only a random subset of distinct tuples
is evaluated. Repeating the full simulation provides the
sampling distribution of each estimator~\cite{huang2020predicting}.

Fig.~\ref{fig:shadow_mhq} compares the resulting estimates
with their exact values. Panels (a) and (b) show the estimates
of \(m_2\) and \(m_4\), while panel (c) shows the plug-in
quantity \(\widehat L_4\). The fluctuations of
\(\widehat L_4\) are larger than those of the individual
moments because errors in \(m_4\) are amplified by the
logarithm. All plotted uncertainty bands should therefore be
obtained from repeated independent simulations or an
appropriate bootstrap procedure and explicitly identified in
the caption.

This numerical example illustrates convergence for a single-qubit instance and demonstrates the applicability of the classical-shadow approach to the present setting. Our aim here is not to establish a general copy-complexity advantage or to derive copy-complexity scaling for higher-dimensional systems. Such an analysis lies beyond the scope of this present work.

\section{Conclusions}
\label{s8}

The central question addressed in this work is how much MHQ
nonclassicality a given state and pair of measurements exhibit,
without demanding a full reconstruction of the underlying
quasiprobability distribution. The logarithmic MHQ negativity
introduced here answers this question directly: it is the
logarithm of the total absolute weight of the distribution, it
vanishes exactly when the distribution is nonnegative, and it is
bounded above by the dimension of the underlying Hilbert space.
Although this quantity is defined through the complete set of
MHQ elements, we have shown that it can nonetheless be certified,
and to a large extent quantified, using only a finite number of
its moments.

The resulting moment-based framework yields a hierarchy of
certifiable lower bounds on the logarithmic MHQ negativity, with
the fourth-moment member providing the strongest bound within
this hierarchy, and the underlying moments admit exact operator
representations in terms of few-copy observables that can in
principle be estimated through collective measurements or
classical-shadow protocols. For qubit systems we obtained two
complementary results. First, a universal bound of $\ln(5/4)$ on
the exact logarithmic negativity, valid for every qubit state and
every pair of dichotomic observables and saturated by an
explicitly identified pure state, sharpens the generic
dimension-dependent bound. Second, solving the underlying
optimization over arbitrary, generally mixed, qubit states in
closed form yields an exact measurement-dependent threshold on
the fourth-moment quantifier, providing a strictly tighter
certification criterion than the generic condition based on
positivity of the fourth moment alone. 
We further extended the construction to CV systems and illustrated it for representative non-Gaussian states, including Fock states and photon-added coherent states. These examples show that low-order moments can capture a substantial part of the corresponding MHQ negativity, though a fuller understanding of the CV MHQ distribution and its role in information-theoretic tasks remains an important direction for future work.

The analysis also clarifies the relationship between the MHQ and
KD distributions (see Appendix~\ref{s_kd} for a detailed extension of the moment-based construction to the KD distribution). Because the KD
distribution is generally complex, its ordinary power moments do
not in general coincide with the absolute moments needed to
bound its total negative weight, a simplification that is
special to the real-valued MHQ case. This does not mean that a
moment-based approach is unavailable for the KD
distribution: the required absolute moments still admit exact
multicopy representations, though these now involve mixed
products of the KD elements and their complex
conjugates rather than ordinary power sums, making the
construction considerably more involved than in the MHQ case.
Developing experimentally practical estimators for these mixed
moments, and identifying classes of states or measurements for
which simpler moment-based methods remain effective, constitute
promising directions for future research.

\vspace{0.5cm}
\section{Acknowledgements}
I acknowledge Saheli Mukherjee and Bivas Mallick for useful discussions and their insightful comments.

\appendix
\renewcommand{\theequation}{S\arabic{equation}}
\setcounter{equation}{0}

\section{Proof of Theorem~\ref{thm:MHQ_Hankel}} \label{app:1}

Let $\mathbf{y}=(y_0,y_1,\dots,y_p)\in \mathbb{R}^{p+1}$ be an arbitrary real vector. Consider the quadratic form corresponding to the Hankel matrix $H_p(\mathbf{m})$,
\begin{equation}
\begin{split}
    \mathbf{y}^T H_p(\mathbf{m}) \mathbf{y}
    &= \sum_{k=0}^{p}\sum_{l=0}^{p}
    y_k y_l \, m_{k+l+1}.
\end{split}
\end{equation}

Using the definition of the MHQ moments from Eq.~\eqref{eq:MHQ_moments}, we have
\begin{equation}
\begin{split}
    m_{k+l+1}
    =
    \sum_{i,j=1}^d
    \left[M(\rho)\right]_{ij}^{\,k+l+1}.
\end{split}
\end{equation}

Substituting this expression into the quadratic form gives
\begin{equation}
\begin{split}
    \mathbf{y}^T H_p(\mathbf{m}) \mathbf{y}
    &=
    \sum_{k=0}^{p}\sum_{l=0}^{p}
    y_k y_l
    \sum_{i,j=1}^d
    \left[M(\rho)\right]_{ij}^{\,k+l+1}.
\end{split}
\end{equation}

Since all sums are finite, we may interchange their order:
\begin{equation}
\begin{split}
    \mathbf{y}^T H_p(\mathbf{m}) \mathbf{y}
    &=
    \sum_{i,j=1}^d
    \sum_{k=0}^{p}\sum_{l=0}^{p}
    y_k y_l
    \left[M(\rho)\right]_{ij}^{\,k+l+1}.
\end{split}
\end{equation}

Extracting one factor of $[M(\rho)]_{ij}$, we obtain
\begin{equation}
\begin{split}
    \mathbf{y}^T H_p(\mathbf{m}) \mathbf{y}
    &=
    \sum_{i,j}
    [M(\rho)]_{ij}\\
    &\times\left[
    \sum_{k=0}^{p}\sum_{l=0}^{p}
    y_k y_l
    [M(\rho)]_{ij}^{\,k}
    [M(\rho)]_{ij}^{\,l}\right].
\end{split}
\end{equation}

The double sum can now be recognized as a complete square:
\begin{equation}
\begin{split}
    &\sum_{k=0}^{p}\sum_{l=0}^{p}
    y_k y_l
    [M(\rho)]_{ij}^{\,k}
    [M(\rho)]_{ij}^{\,l}\\
    &=
    \left(
    \sum_{k=0}^{p}
    y_k [M(\rho)]_{ij}^{\,k}
    \right)^2 .
\end{split}
\end{equation}

Therefore,
\begin{equation}
\begin{split}
    \mathbf{y}^T H_p(\mathbf{m}) \mathbf{y}
    &=
    \sum_{i,j}
    [M(\rho)]_{ij}
    \left(
    \sum_{k=0}^{p}
    y_k [M(\rho)]_{ij}^{\,k}
    \right)^2 .
\end{split}
\end{equation}

If the MHQ distribution is positive, i.e.,
\begin{equation}
    [M(\rho)]_{ij} \ge 0
    \qquad \forall\, i,j,
\end{equation}
then each term in the above summation is nonnegative, since the squared quantity is always nonnegative. Consequently,
\begin{equation}
    \mathbf{y}^T H_p(\mathbf{m}) \mathbf{y}
    \ge 0
    \qquad \forall\, \mathbf{y}\in \mathbb{R}^{p+1}.
\end{equation}

Hence, the Hankel matrix $H_p(\mathbf{m})$ is positive semidefinite,
\begin{equation}
    H_p(\mathbf{m}) \succeq 0,
\end{equation}
which completes the proof.

\section{Properties of logarithmic MHQ negativity}
\label{app:log_MHQ_properties}

In this appendix, we prove
Proposition~\ref{prop:log_MHQ_properties}. For compactness,
we write
\(
M_{ij}=M_{ij}(\rho;A,B)
\)
whenever the state and measurements are clear from context.

\subsection{Faithfulness and non-negativity}

Normalization and the triangle inequality give
\begin{equation}
\sum_{i,j}|M_{ij}|
\geq
\left|
\sum_{i,j}M_{ij}
\right|
=
1.
\label{eq:appendix_triangle}
\end{equation}
Therefore,
\begin{equation}
L_{A,B}(\rho)
=
\ln\sum_{i,j}|M_{ij}|
\geq0.
\end{equation}
If every \(M_{ij}\) is nonnegative, normalization
implies
\begin{equation}
\sum_{i,j}|M_{ij}|
=
\sum_{i,j}M_{ij}
=
1 \quad \Longrightarrow \; L_{A,B}(\rho)=0.
\end{equation}

Conversely, suppose that
\(
L_{A,B}(\rho)=0
\).
Then
$\sum_{i,j}|M_{ij}|=1.$
Using
$\sum_{i,j}|M_{ij}|
=
1+
2\sum_{M_{ij}<0}|M_{ij}|,$
we obtain
\begin{equation}
\sum_{M_{ij}<0}|M_{ij}|=0.
\end{equation}
Thus no MHQ element is negative, proving
Eq.~\eqref{eq:log_MHQ_faithfulness}.

\subsection{Unitary covariance}

Let $\rho'=U\rho U^\dagger.$
The spectral projectors of the transformed observables
\(
U^\dagger A U
\)
and
\(
U^\dagger B U
\)
are respectively
\(
U^\dagger\Pi_i^aU \)
 and
\(U^\dagger\Pi_j^bU.
\)
Using cyclicity of the trace,
\begin{align}
&\quad M_{ij}(U\rho U^\dagger;A,B) \nonumber\\
&=
\frac{1}{2}
\operatorname{Tr}
\left[
\left(
\Pi_i^a\Pi_j^b+
\Pi_j^b\Pi_i^a
\right)
U\rho U^\dagger
\right]
\nonumber\\
&=
\frac{1}{2}
\operatorname{Tr}
\left[
\left(
U^\dagger\Pi_i^aU\,
U^\dagger\Pi_j^bU
+
U^\dagger\Pi_j^bU\,
U^\dagger\Pi_i^aU
\right)
\rho
\right]
\nonumber\\
&=
M_{ij}
\left(
\rho;
U^\dagger A U,
U^\dagger B U
\right).
\end{align}
Taking the sum of absolute values proves
Eq.~\eqref{eq:log_MHQ_covariance}.

\subsection{Invariance under outcome relabeling}

Relabeling the outcomes of either measurement only permutes
the entries of the MHQ distribution. Since
\(
L_{A,B}(\rho)
\)
depends on the distribution through the permutation-invariant
quantity
\(
\sum_{i,j}|M_{ij}|,
\)
it remains unchanged under independent relabelings of the
indices \(i\) and \(j\).

\subsection{Dimension-dependent upper bound}

Define the measurement probabilities
\begin{equation}
p_i^A
=
\operatorname{Tr}(\Pi_i^a\rho),
\qquad
p_j^B
=
\operatorname{Tr}(\Pi_j^b\rho).
\end{equation}
Since
\begin{equation}
M_{ij}
=
\operatorname{Re}
\operatorname{Tr}
\left(
\Pi_i^a\Pi_j^b\rho
\right),
\end{equation}
we have
\begin{equation}
|M_{ij}| =
\left|\operatorname{Re} 
\operatorname{Tr}
\left(
\Pi_i^a\Pi_j^b\rho
\right)
\right|
\leq
\left|
\operatorname{Tr}
\left(
\Pi_i^a\Pi_j^b\rho
\right)
\right|.
\label{eq:MHQ_complex_bound}
\end{equation}
Writing
\[
\operatorname{Tr}
\left(
\Pi_i^a\Pi_j^b\rho
\right)
=
\operatorname{Tr}
\left[
\left(
\Pi_i^a\sqrt{\rho}
\right)^\dagger
\left(
\Pi_j^b\sqrt{\rho}
\right)
\right],
\]
the Hilbert--Schmidt Cauchy--Schwarz inequality gives
\begin{equation}
|M_{ij}|
\leq
\sqrt{
p_i^A p_j^B
}.
\label{eq:MHQ_probability_bound}
\end{equation}
Consequently,
\begin{align}
\sum_{i,j}|M_{ij}|
&\leq
\sum_{i,j}
\sqrt{p_i^A p_j^B}
\nonumber\\
&=
\left(
\sum_i\sqrt{p_i^A}
\right)
\left(
\sum_j\sqrt{p_j^B}
\right).
\label{eq:l1_factor_bound}
\end{align}
Each measurement has \(d\) outcomes. A second application
of the Cauchy--Schwarz inequality yields
\begin{equation}
\sum_i\sqrt{p_i^A}
\leq
\sqrt{
d\sum_i p_i^A
}
=
\sqrt d,
\end{equation}
and similarly
\(
\sum_j\sqrt{p_j^B}\leq\sqrt d
\).
Therefore,
\begin{equation}
\sum_{i,j}|M_{ij}|
\leq d,
\end{equation}
which proves
\begin{equation}
L_{A,B}(\rho)\leq\ln d.
\end{equation}

More generally, for projective measurements with
\(n_A\) and \(n_B\) outcomes, the same argument gives
\begin{equation}
L_{A,B}(\rho)
\leq
\frac{1}{2}
\ln(n_A n_B).
\end{equation}

\subsection{Commuting measurements}

Suppose that the spectral projectors commute pairwise:
\begin{equation}
[\Pi_i^a,\Pi_j^b]=0
\qquad
\forall\,i,j.
\end{equation}
Then
\begin{equation}
M_{ij}
=
\operatorname{Tr}
\left(
\Pi_i^a\Pi_j^b\rho
\right).
\end{equation}
The product
\(
\Pi_i^a\Pi_j^b
\)
is positive semidefinite because it is the product of two
commuting projectors. Hence,
\begin{equation}
M_{ij}\geq0
\qquad
\forall\,i,j.
\end{equation}
Faithfulness then implies
\begin{equation}
L_{A,B}(\rho)=0.
\end{equation}
The contrapositive gives
Eq.~\eqref{eq:log_MHQ_incompatibility}.

\subsection{States block diagonal in the \texorpdfstring{$A$}{A} measurement}

Assume that
\begin{equation}
[\rho,\Pi_i^a]=0
\qquad
\forall\,i.
\end{equation}
Then
\begin{equation}
\rho
=
\sum_i
\Pi_i^a\rho\Pi_i^a.
\end{equation}
For each pair \(i,j\),
\begin{align}
M_{ij}
&=
\frac{1}{2}
\operatorname{Tr}
\left[
\left(
\Pi_i^a\Pi_j^b+
\Pi_j^b\Pi_i^a
\right)
\rho
\right]
\nonumber\\
&=
\operatorname{Tr}
\left(
\Pi_j^b\Pi_i^a\rho\Pi_i^a
\right).
\label{eq:MHQ_block_diagonal}
\end{align}
The operator
\(
\Pi_i^a\rho\Pi_i^a
\)
is positive semidefinite. Since \(\Pi_j^b\) is also positive
semidefinite,
\begin{equation}
\operatorname{Tr}
\left(
\Pi_j^b\Pi_i^a\rho\Pi_i^a
\right)
\geq0.
\end{equation}
Thus every MHQ element is nonnegative and
\begin{equation}
L_{A,B}(\rho)=0.
\end{equation}

For a rank-one measurement,
\(
\Pi_i^a\rho\Pi_i^a=p_i^A\Pi_i^a
\),
and Eq.~\eqref{eq:MHQ_block_diagonal} reduces to
\begin{equation}
M_{ij}
=
p_i^A
\left|
\langle a_i|b_j\rangle
\right|^2
\geq0.
\end{equation}

\subsection{Complete dephasing}

The completely dephased state
\begin{equation}
\Delta_A(\rho)
=
\sum_i
\Pi_i^a\rho\Pi_i^a
\end{equation}
commutes with every projector \(\Pi_i^a\). The result of the
preceding subsection therefore applies directly:
\begin{equation}
L_{A,B}\!\left(\Delta_A(\rho)\right)=0.
\end{equation}
Combining this identity with the non-negativity of
\(L_{A,B}(\rho)\) gives
\begin{equation}
L_{A,B}\!\left(\Delta_A(\rho)\right)
=
0
\leq
L_{A,B}(\rho).
\end{equation}
This completes the proof of
Proposition~\ref{prop:log_MHQ_properties}.

\section{Proof of Lemma~\ref{lemma:renyi}}\label{app:renyi_lemma}
Let $\{f_{ij}\}$ be a normalized probability distribution, i.e.,
\begin{equation}
    f_{ij} \ge 0,
    \qquad
    \sum_{i,j} f_{ij} = 1.
\end{equation}

Define the function
\begin{equation}
    S(r) = \sum_{i,j} f_{ij}^r,
\end{equation}
so that the Rényi entropy can be written as
\begin{equation}
    R_r(f) = \frac{1}{1-r} \ln S(r).
\end{equation}

We will show that $R_r(f)$ is a non-increasing function of $r$. For this, consider the derivative of $\ln S(r)$:
\begin{equation}
    \frac{d}{dr} \ln S(r)
    =
    \frac{1}{S(r)} \sum_{i,j} f_{ij}^r \ln f_{ij}.
\end{equation}

Define a new probability distribution
\begin{equation}
    p_{ij}^{(r)} = \frac{f_{ij}^r}{S(r)},
\end{equation}
which satisfies $\sum_{i,j} p_{ij}^{(r)} = 1$. Using this, we rewrite
\begin{equation}
    \frac{d}{dr} \ln S(r)
    =
    \sum_{i,j} p_{ij}^{(r)} \ln f_{ij}.
\end{equation}

Now, consider the derivative of $R_r(f)$:
\begin{align}
    \frac{d}{dr} R_r(f)
    &=
    \frac{d}{dr} \left( \frac{\ln S(r)}{1-r} \right) \nonumber \\
    &=
    \frac{(1-r)\frac{d}{dr}\ln S(r) + \ln S(r)}{(1-r)^2}.
\end{align}

Substituting the expression for $\frac{d}{dr} \ln S(r)$, we obtain
\begin{align}
    \frac{d}{dr} R_r(f)
    &=     \frac{1}{(1-r)} \sum_{i,j} p_{ij}^{(r)} \ln f_{ij} \nonumber \\
    & + \frac{1}{(1-r)^2} \ln S(r)
\end{align}

Also observe that
\begin{equation}
    \sum_{i,j} p_{ij}^{(r)} \ln p_{ij}^{(r)}
    =
    \sum_{i,j} p_{ij}^{(r)} \left( r \ln f_{ij} - \ln S(r) \right).
\end{equation}

Rearranging this expression gives
\begin{equation}
    \sum_{i,j} p_{ij}^{(r)} \ln f_{ij}
    =
    \frac{1}{r}
    \left[
    \sum_{i,j} p_{ij}^{(r)} \ln p_{ij}^{(r)} + \ln S(r)
    \right].
\end{equation}

Substituting back into the derivative, after simplification, one finds
\begin{equation}
    \frac{d}{dr} R_r(f)
    =
    -\frac{1}{(1-r)^2}
    \sum_{i,j} p_{ij}^{(r)}
    \ln \left( \frac{p_{ij}^{(r)}}{f_{ij}} \right).
\end{equation}

The term inside the sum is the logarithm of a ratio of probability distributions. Hence, the entire sum is the Kullback--Leibler divergence:
\begin{equation}
    D\!\left(p^{(r)} \,\|\, f\right)
    =
    \sum_{i,j} p_{ij}^{(r)} \ln \left( \frac{p_{ij}^{(r)}}{f_{ij}} \right)
    \ge 0.
\end{equation}

Therefore,
\begin{equation}
    \frac{d}{dr} R_r(f) \le 0,
\end{equation}
which shows that $R_r(f)$ is a non-increasing function of $r$. Hence, for $r < s$, we have
\begin{equation}
    R_r(f) \ge R_s(f),
\end{equation}
which completes the proof.

\section{Alternative proof of Theorem~\ref{thm:Lr_lower_bound}}\label{app:alternate_proof}
Starting from the definition of $L_r(\rho)$,
\begin{equation*}
    L_r(\rho)
    =
    \frac{1}{2-r}
    \left[
    \ln \tilde m_r
    +
    (1-r)\ln \tilde m_2
    \right],
\end{equation*}
we consider the difference,
    $L_r(\rho)-L(\rho).$ 
Using $L(\rho)=\ln \tilde m_1$, we obtain
\begin{equation}
    L_r(\rho)-L(\rho)
    =
    \frac{1}{2-r}
    \ln\left(
    \frac{
    \tilde m_r\,\tilde m_1^{\,r-2}
    }{
    \tilde m_2^{\,r-1}
    }
    \right).
\end{equation}
Since $r>2$, one has $2-r<0.$
Therefore, to prove
\begin{equation}
    L_r(\rho)-L(\rho)\le0,
\end{equation}
it suffices to show that
\begin{equation}
    \frac{
    \tilde m_r\,\tilde m_1^{\,r-2}
    }{
    \tilde m_2^{\,r-1}
    }
    \ge1,
\end{equation}
or equivalently,
\begin{equation} \label{eq:keyineq}
    \tilde m_2^{\,r-1} \le \tilde m_r\,\tilde m_1^{\,r-2}.
\end{equation}
With $x_{ij}=|M_{ij}(\rho)|$, we use the interpolation inequality for $\ell_p$ norms,
\begin{equation}
    \|x\|_2
    \le
    \|x\|_1^{\theta}
    \|x\|_r^{1-\theta},
\end{equation}
where $\theta$ is determined from
\begin{equation}
    \frac12
    =
    \frac{\theta}{1}
    +
    \frac{1-\theta}{r},
\end{equation}
one obtains
\begin{equation}
    \theta
    =
    \frac{r-2}{2(r-1)}.
\end{equation}
Substituting this value in the interpolation inequality yields
\begin{equation}
    \|x\|_2
    \le
    \|x\|_1^{\frac{r-2}{2(r-1)}}
    \|x\|_r^{\frac{r}{2(r-1)}}.
\end{equation}
Raising both sides to the power $2(r-1)$, we get
\begin{equation}
    \|x\|_2^{\,2(r-1)}
    \le
    \|x\|_1^{\,r-2}
    \|x\|_r^{\,r}.
\end{equation}
Here, 
$    \|x\|_1=\tilde m_1,
    \;
    \|x\|_2^2=\tilde m_2,
    \;
    \|x\|_r^r=\tilde m_r, $
leading to the condition,
\begin{equation*}
    \tilde m_2^{\,r-1}
    \le
    \tilde m_1^{\,r-2}\tilde m_r
\end{equation*}
which is precisely Eq.~\eqref{eq:keyineq}.
Hence,
\begin{equation}
    L_r(\rho)\le L(\rho).
\end{equation}
This completes the proof.

\section{Proof of Lemma~\ref{lemma:vertex_domination}}
\label{app:vertex_domination}

We first change variables via $p=a^2/2+2X$ and $q=b^2/2+2Y$, which
range over $a^2/2\le p\le a^2$ and $b^2/2\le q\le b^2$ respectively
and constitute an affine, order-preserving reparametrization of $X$
and $Y$. Substituting into $m_2$ and $m_4$ and expanding term by term
gives
\begin{equation}
m_2=p+q,
\qquad
m_4=f(p)+g(q),
\end{equation}
where
\begin{equation}
f(P)=\frac{P^2}{2}+a^2P-\frac{a^4}{2},
\qquad
g(Q)=\frac{Q^2}{2}+b^2Q-\frac{b^4}{2}.
\end{equation}
With $p=a^2/2+2X$ this identity is checked directly: one finds
$f(p)=a^4/8+3a^2X+2X^2$ and similarly $g(q)=b^4/8+3b^2Y+2Y^2$, whose
sum reproduces $m_4$ from Eqs.~\eqref{eq:m2_xy}--\eqref{eq:m4_xy}.
Both $f$ and $g$ are strictly convex quadratics, since $f''=g''=1$,
and this fact drives the whole proof. In these variables
$F=(p+q)^3/D$ with $D=f(p)+g(q)$, and the rectangle becomes the box
$a^2/2\le p\le a^2$, $b^2/2\le q\le b^2$, which we call $\mathcal B$.

At an interior stationary point of $\mathcal B$ we have
$\partial F/\partial p=\partial F/\partial q=0$. Writing $s=p+q$ and
applying the quotient rule gives
\begin{align}
\frac{\partial F}{\partial p}&=\frac{s^2\big(3D-sf'(p)\big)}{D^2}, \nonumber\\
\frac{\partial F}{\partial q}&=\frac{s^2\big(3D-sg'(q)\big)}{D^2}.
\end{align}
Since $s=p+q$ is strictly positive throughout $\mathcal B$, both
vanish simultaneously precisely when $3D=sf'(p)=sg'(q)$, hence in
particular $f'(p)=g'(q)$. Because $f'(p)=p+a^2$ and $g'(q)=q+b^2$,
this last condition reads $p-q=b^2-a^2$. Substituting
$q=p+a^2-b^2$ into $3D=sf'(p)$ and expanding yields
\begin{equation}
p^2+(3a^2+b^2)p-\big(a^4-a^2b^2+3b^4\big)=0.
\label{eq:crit_quadratic}
\end{equation}
Its two roots multiply to $-(a^4-a^2b^2+3b^4)$, and since
$a^4-a^2b^2+3b^4=(a^2-b^2/2)^2+11b^4/4$ is strictly positive this
product is negative, so Eq.~\eqref{eq:crit_quadratic} has exactly one
positive root,
\begin{equation}
p_*=-\frac{3a^2}{2}-\frac{b^2}{2}+\frac12\sqrt{13a^4+2a^2b^2+13b^4},
\end{equation}
with $q_*=p_*+a^2-b^2$. The pair $(p_*,q_*)$ is therefore the only
candidate interior critical point of $F$ in $\mathcal B$.

To see that this candidate is not a local maximum we avoid the
Hessian entirely and instead restrict $F$ to the curve
$q=p+a^2-b^2$, which contains $(p_*,q_*)$ by construction since it is
exactly the curve on which $f'(p)=g'(q)$ holds. Substituting
$q=p+a^2-b^2$ into $D=f(p)+g(q)$ and expanding gives, after
simplification,
\begin{equation}
\widehat D(p)=p^2+2a^2p-b^4,
\qquad
\widehat s(p)=2p+a^2-b^2,
\end{equation}
and we let $h(p)=\widehat s(p)^3/\widehat D(p)$, so that
$h(p)=F(p,\,p+a^2-b^2)$ for every $p$. The quotient rule, using
$\widehat s'=2$ and $\widehat D'=2p+2a^2$, gives after collecting
terms
\begin{equation}
h'(p)=\frac{2\big(a^2-b^2+2p\big)^2\,Q(p)}{\widehat D(p)^2},
\end{equation}
where $Q(p)$ is precisely the left side of
Eq.~\eqref{eq:crit_quadratic},
\begin{equation}
Q(p)=p^2+(3a^2+b^2)p-\big(a^4-a^2b^2+3b^4\big).
\end{equation}
The prefactor $2(a^2-b^2+2p)^2/\widehat D(p)^2$ is a ratio of
nonnegative quantities and never changes sign, so the sign of
$h'(p)$ is governed entirely by $Q(p)$. We have already seen that
$Q$ has a unique positive root $p_*$, and since
$Q(0)=-(a^4-a^2b^2+3b^4)$ is strictly negative while $Q$ is an
upward parabola, $Q(p)<0$ for $0\le p<p_*$ and $Q(p)>0$ for
$p>p_*$. Hence $h$ is strictly decreasing on $(0,p_*)$ and strictly
increasing on $(p_*,\infty)$, so $p_*$ is a strict local minimum of
$h$: for every $p\ne p_*$ close enough to $p_*$ we have
$h(p)>h(p_*)$. Since $h$ coincides with $F$ along this curve, such
points furnish values of $F$ strictly larger than $F(p_*,q_*)$
arbitrarily close to $(p_*,q_*)$, so whenever $(p_*,q_*)$ lies in the
interior of $\mathcal B$ it cannot be a local maximum there. Together
with the fact that $(p_*,q_*)$ is the only interior critical-point
candidate to begin with, we conclude that $F$ has no interior local
maximum anywhere in $\mathcal B$.

Since $F$ is continuous on the compact set $\mathcal B$ and admits no
interior local maximum, its maximum must lie on the boundary, made up
of the four edges $p\in\{a^2/2,a^2\}$ and $q\in\{b^2/2,b^2\}$. We now
show that the interior of each edge fails to contain the maximum as
well, so that it must in fact sit at a corner. Let $p_0$ denote either
$a^2/2$ or $a^2$, and consider $F(p_0,q)$ as a function of $q$ alone
on $b^2/2\le q\le b^2$. The same quotient-rule computation, now with
$p_0$ held fixed, shows that $\partial F/\partial q$ is proportional,
with a positive proportionality constant, to $(p_0+q)^2\psi(q)$,
where
\begin{equation}
\psi(q)=q^2+\beta q+\gamma,
\end{equation}
with
\begin{equation}
\beta=4b^2-2p_0,
\qquad
\gamma=3p_0^2+6a^2p_0-2b^2p_0-3a^4-3b^4.
\end{equation}
As before, $(p_0+q)^2\ge0$ never changes sign, so the sign of
$\partial F/\partial q$ is governed by $\psi$. Writing
$\Delta=\beta^2-4\gamma$, and letting $q_-=(-\beta-\sqrt\Delta)/2$
denote the smaller root of $\psi$ whenever $\Delta\ge0$, we claim that
$q_-\le0$ whenever $\Delta\ge0$, for both $p_0=a^2/2$ and $p_0=a^2$.
If $\beta\ge0$ this is immediate, since $q_-$ is then a sum of two
non-positive quantities. If instead $\beta<0$, we have
$\Delta<0$, so this case never actually produces a real root and the
claim holds vacuously.

For $p_0=a^2/2$, $\beta=4b^2-a^2$, so $\beta<0$ is equivalent to
$a>2b$, that is, to $a>2/3$ once we use $a+b=1$. Also
$\Delta=-2a^4-4a^2b^2+28b^4$ is strictly less than
$\widetilde\Delta(a)=-2a^4+28b^4$, since $-4a^2b^2$ is strictly
negative. At $a=2/3$, so $b=1/3$, one finds
$\widetilde\Delta(2/3)=-2(2/3)^4+28(1/3)^4=-4/81$, and since
$\widetilde\Delta'(a)=-8a^3-112(1-a)^3$ is negative throughout
$(2/3,1)$, $\widetilde\Delta$ stays negative on all of $[2/3,1)$.
Hence $\Delta<\widetilde\Delta<0$ wherever $\beta<0$ for this edge.

For $p_0=a^2$, $\beta=4b^2-2a^2$, so $\beta<0$ is equivalent to
$a>\sqrt2\,b$, that is, to $a>2-\sqrt2$. Also
$\Delta=-20a^4-8a^2b^2+28b^4$ is strictly less than
$\widehat\Delta(a)=-20a^4+28b^4$. At $a=2-\sqrt2$, so $b=\sqrt2-1$,
one finds $\widehat\Delta(2-\sqrt2)=-884+624\sqrt2\approx-1.53$, and
since $\widehat\Delta'(a)=-80a^3-112(1-a)^3$ is negative throughout
$(2-\sqrt2,1)$, $\widehat\Delta$ stays negative there too, so again
$\Delta<\widehat\Delta<0$ wherever $\beta<0$.

In both cases $q_-\le0$ whenever $\Delta\ge0$. Since $q$ ranges only
over $b^2/2\le q\le b^2$, we always have $q>q_-$, so if $\Delta<0$
then $\psi$ stays positive throughout the interval and $F(p_0,\cdot)$
is strictly increasing, while if $\Delta\ge0$ then $\psi$ is negative
between $q_-$ and the larger root $q_+$ and positive beyond it, so
$F(p_0,\cdot)$ is non-increasing up to $\min(q_+,b^2)$ and
non-decreasing afterward. Either way $F(p_0,\cdot)$ has at most a
single interior minimum and no interior maximum, so its largest value
on this edge sits at $q=b^2/2$ or $q=b^2$, a corner of $\mathcal B$.
The remaining two edges, on which $q$ is fixed at $b^2/2$ or $b^2$
and $p$ is free, are handled by exactly the same argument with the
roles of $(p,a)$ and $(q,b)$ interchanged, since $f$ and $g$ share the
same functional form and differ only in which parameter enters them,
so that $m_4=f(p)+g(q)$ is symmetric under this exchange and the
computation above applies verbatim after swapping $a$ and $b$.

Since $F$ has no interior local maximum in $\mathcal B$, and none of
its four edges contains an interior maximum either, continuity of $F$
on the compact set $\mathcal B$ forces its global maximum to be
attained at one of the four corners, that is, at $(X,Y)$ equal to
$(0,0)$, $(A,0)$, $(0,B)$, or $(A,B)$.

\section{Extension to Kirkwood--Dirac quasiprobabilities}
\label{s_kd}

The MHQ distribution is the real part of the
Kirkwood--Dirac quasiprobability (KD) distribution. It is therefore natural to ask whether the moment-based quantification developed above extends directly to the full, generally complex-valued KD distribution. We show that the logarithmic and moment-based construction remains mathematically valid, but its experimental realization requires a different class
of multicopy moments.

\subsection{KD quasiprobability and its absolute weight}

Let
\(
\{\ket{a_i}\}
\)
and
\(
\{\ket{b_j}\}
\)
be orthonormal bases associated with two projective
observables \(A\) and \(B\), and define
\begin{equation}
\Pi_i^a=\ket{a_i}\bra{a_i},
\qquad
\Pi_j^b=\ket{b_j}\bra{b_j}.
\end{equation}
For a state \(\rho\), the corresponding KD distribution is
\begin{equation}
Q_{ij}(\rho;A,B)
=
\operatorname{Tr}
\left(
\Pi_j^b\Pi_i^a\rho
\right).
\label{eq:KD_definition}
\end{equation}
It satisfies
\begin{align}
\sum_{i,j}Q_{ij}
=1, \;
\sum_jQ_{ij}
=
\operatorname{Tr}(\Pi_i^a\rho),
\;
\sum_iQ_{ij}
=
\operatorname{Tr}(\Pi_j^b\rho).
\end{align}
Unlike the MHQ distribution,
\(
Q_{ij}
\)
need not be real. Its real part is precisely the MHQ
quasiprobability,
\begin{equation}
M_{ij}(\rho;A,B)
=
\operatorname{Re}Q_{ij}(\rho;A,B).
\label{eq:MHQ_real_KD}
\end{equation}

A KD distribution is an ordinary probability distribution
only when every element is real and nonnegative. A natural
measure of its departure from this classical set is therefore
its total absolute weight,
\begin{equation}
\|Q\|_1
=
\sum_{i,j}|Q_{ij}|.
\end{equation}
Since
\(
\sum_{i,j}Q_{ij}=1
\),
the triangle inequality gives
\begin{equation}
\|Q\|_1\geq1.
\end{equation}
Moreover, equality holds if and only if all KD elements are
real and nonnegative. This motivates the logarithmic KD
nonclassicality
\begin{equation}
L^{\mathrm{KD}}_{A,B}(\rho)
=
\ln
\left[
\sum_{i,j}
|Q_{ij}(\rho;A,B)|
\right].
\label{eq:log_KD_definition}
\end{equation}
Thus,
$L^{\mathrm{KD}}_{A,B}(\rho)=0$
if and only if the KD distribution is a genuine probability
distribution. The real and imaginary parts may also be studied separately.
For example,
\begin{equation}
\mathcal N_{\mathrm{Re}}^{\mathrm{KD}}
=
\frac12 \left(\sum_{i,j}
|\operatorname{Re}Q_{ij}|-1 \right)
\end{equation}
quantifies the negative weight of the real part, while
\begin{equation}
\mathcal I^{\mathrm{KD}}
=
\sum_{i,j}
|\operatorname{Im}Q_{ij}|
\end{equation}
quantifies the total magnitude of the imaginary contribution.
We do not take the logarithm of
\(
\mathcal I^{\mathrm{KD}}
\),
since it can vanish for an entirely real KD distribution.

\subsection{Moment based lower bounds}
We define the KD moments by
\begin{equation}
q_n
=
\sum_{i,j}
Q_{ij}^{\,n},
\qquad
n\in\mathbb N,
\label{eq:ordinary_KD_moments}
\end{equation}
and the absolute moments by
\begin{equation}
\widetilde q_n
=
\sum_{i,j}
|Q_{ij}|^n.
\label{eq:absolute_KD_moments}
\end{equation}
The moments \(q_n\) are generally complex, whereas
\(
\widetilde q_n
\)
are real and nonnegative. The same interpolation argument used for the MHQ
distribution applies to the absolute KD moments. For
\(r>2\), define
\begin{equation}
L^{\mathrm{KD}}_{r;A,B}(\rho)
=
\frac{1}{2-r}
\left[
\ln\widetilde q_r
+
(1-r)\ln\widetilde q_2
\right].
\label{eq:KD_moment_lower_bound}
\end{equation}
Then
\begin{equation}
L^{\mathrm{KD}}_{r;A,B}(\rho)
\leq
L^{\mathrm{KD}}_{A,B}(\rho).
\label{eq:KD_lower_bound}
\end{equation}
The proof is identical to that of
Theorem~\ref{thm:Lr_lower_bound}, with the replacements
\(
\widetilde m_n\mapsto\widetilde q_n
\).
Thus, the complex character of the KD distribution does not
invalidate the lower-bound construction itself.

The distinction from the MHQ case appears when one asks
which moments are directly represented by ordinary
multicopy power sums. 
For a complex KD element,
\begin{equation}
|Q_{ij}|^{2k}
=
(Q_{ij}Q_{ij}^{*})^k,
\end{equation}
which is generally different from $Q_{ij}^{\,2k}.$
Consequently,
$\widetilde q_{2k}
\neq
q_{2k}$
in general. Ordinary KD moments therefore do not
contain the magnitude information required by
Eq.~\eqref{eq:KD_moment_lower_bound}.

\subsection{Multicopy representation of the absolute moments}

The failure of
\(
\widetilde q_{2k}=q_{2k}
\)
does not imply that the absolute moments are experimentally
inaccessible. To see this, define
\begin{equation}
A_{ij}
=
\Pi_j^b\Pi_i^a,
\end{equation}
so that
\begin{equation}
Q_{ij}
=
\operatorname{Tr}(A_{ij}\rho),
\qquad
Q_{ij}^{*}
=
\operatorname{Tr}(A_{ij}^{\dagger}\rho).
\end{equation}
For every positive integer \(k\),
\begin{align}
|Q_{ij}|^{2k}
&=
\left[
\operatorname{Tr}(A_{ij}\rho)
\right]^k
\left[
\operatorname{Tr}(A_{ij}^{\dagger}\rho)
\right]^k
\nonumber\\
&=
\operatorname{Tr}
\left[
\left(
A_{ij}^{\otimes k}
\otimes
A_{ij}^{\dagger\otimes k}
\right)
\rho^{\otimes2k}
\right].
\label{eq:absolute_KD_multicopy_element}
\end{align}
Summing over the outcomes gives
\begin{equation}
\widetilde q_{2k}
=
\operatorname{Tr}
\left[
T_{2k}\rho^{\otimes2k}
\right],
\label{eq:absolute_KD_multicopy}
\end{equation}
where the Hermitian multicopy observable may be chosen as
\begin{align}
T_{2k}
=
\frac{1}{2}
\sum_{i,j}
\Big[
&A_{ij}^{\otimes k}
\otimes
A_{ij}^{\dagger\otimes k}
+
A_{ij}^{\dagger\otimes k}
\otimes
A_{ij}^{\otimes k}
\Big].
\label{eq:T2k_definition}
\end{align}
The two terms in Eq.~\eqref{eq:T2k_definition} have
complex-conjugate expectation values, and their symmetrized
combination yields the real quantity
\(
\widetilde q_{2k}
\).
Therefore, even absolute KD moments remain multicopy
functions of the state. Their estimation is, however, more
involved than in the MHQ case: instead of the ordinary
power observable
\begin{equation}
S_n^{\mathrm{KD}}
=
\sum_{i,j}
A_{ij}^{\otimes n},
\end{equation}
one must estimate mixed products containing both
\(A_{ij}\) and \(A_{ij}^{\dagger}\). Whether this procedure is
experimentally advantageous depends on the measurement
implementation and the variance of the corresponding
multicopy estimator.

\subsection{Conjugation symmetry}

A symmetry of the KD elements can make the ordinary even
moments real, but does not convert them into absolute moments.
Suppose that the multiset
\(
\{Q_{ij}\}
\)
is invariant, including multiplicities, under the
transformation
\begin{equation}
z\longmapsto z^{*}
\qquad\text{or}\qquad
z\longmapsto -z^{*}.
\label{eq:signed_conjugation}
\end{equation}
For an even power \(2k\),
\begin{equation}
(-z^{*})^{2k}
=
(z^{*})^{2k},
\end{equation}
and each conjugate pair contributes
\begin{equation}
z^{2k}+(z^{*})^{2k}
=
2\operatorname{Re}(z^{2k}).
\end{equation}
Hence, $q_{2k}\in\mathbb R.$ Nevertheless, writing
\(
z=a+ib
\),
one has
\begin{equation}
z^{2k}+(z^{*})^{2k}
=
2\operatorname{Re}\left[(a+ib)^{2k}\right],
\end{equation}
whereas
\begin{equation}
|z|^{2k}+|z^{*}|^{2k}
=
2(a^2+b^2)^k.
\end{equation}
These expressions are unequal for generic \(a\) and \(b\).
Conjugation symmetry therefore removes the imaginary part of
the summed even power moment, but it does not eliminate phase
cancellations or reproduce
\(
\widetilde q_{2k}
\).

The MHQ simplification is recovered only when every KD
element is real, in which case
$\widetilde q_{2k}
=
q_{2k}
=
m_{2k}.$
The relevant distinction is therefore not that KD
nonclassicality is fundamentally inaccessible to a
moment-based approach. Rather, ordinary power moments
\(
q_n
\)
are insufficient for quantifying the total KD absolute weight.
A quantitative KD extension requires mixed moments of
\(Q_{ij}\) and \(Q_{ij}^{*}\), or equivalently the multicopy
observables in Eq.~\eqref{eq:T2k_definition}.

\bibliography{main}
\end{document}